\documentclass[12pt]{article}

\usepackage{amsfonts}
\usepackage{amssymb}
\usepackage{amsmath}
\usepackage{amsthm}
\usepackage{bbm}
\usepackage{graphicx}
\usepackage{epstopdf}
\usepackage[margin=.75in]{geometry}
\usepackage[allcolors=black,bookmarks=false,breaklinks,colorlinks=true]{hyperref}
\usepackage{natbib}

\newcommand{\norm}[1]{\left\lVert#1\right\rVert}

\numberwithin{equation}{section}
\renewcommand{\baselinestretch}{1.2}
\allowdisplaybreaks

\newtheorem{asu}{{\sc Assumption}}
\newtheorem{corollary}{{\sc Corollary}}
\newtheorem{thm}{\sc Theorem}

\newtheorem{pro}{\sc Proposition}
\newtheorem{lem}{\sc Lemma}
\newtheorem*{assuB}{\sc Assumption B}
\newtheorem*{assuBtilde}{\sc Assumption $\tilde{\mbox{B}}$}
\newtheorem{definition}{{\sc Definition}}

\newcommand{\be}{\begin{eqnarray}}
\newcommand{\ee}{\end{eqnarray}}
\newcommand{\bes}{\begin{eqnarray*}}
\newcommand{\ees}{\end{eqnarray*}}

\makeatletter
\renewenvironment{proof}[1][\proofname]{%
  \par
  \pushQED{\qed}%
  \normalfont
  \topsep6\p@\@plus6\p@\relax
  \trivlist
  \item[\hskip\labelsep
        \textsc{#1}\@addpunct{.}]%
}{%
  \popQED\endtrivlist\@endpefalse
}
\makeatother

\begin{document}

\title{The realized copula of volatility\thanks{Christensen appreciates funding from the Independent Research Fund Denmark to support this work (DFF 1028–00030B). Zhi Liu’s research is supported by MYRG-GRG2024-00190-FST-UMDF, MYRG-GRG2025-00093-FST, and APAEM Seed Grant in Financial Econometrics from the University of Macau. Potiron acknowledges financial support from the Japanese Society for the Promotion of Science (23H00807).}}
\author{Kim Christensen\thanks{Department of Economics and Business Economics, Aarhus University, Denmark. E-mail: \url{kim@econ.au.dk}} \thanks{Research fellow at the Danish Finance Institute (DFI).} \and Wenjing Liu\thanks{Department of Mathematics, University of Macau, Macau SAR, China. E-mail: \url{yc27481@um.edu.mo}} \and Zhi Liu\thanks{Department of Mathematics, University of Macau, Macau SAR, China. E-mail: \url{liuzhi@um.edu.mo}} \and Yoann Potiron\thanks{Faculty of Business and Commerce, Keio University, Japan. E-mail: \url{potiron@fbc.keio.ac.jp}}}
\date{April, 2026}

\maketitle

\begin{abstract}
We study a new measure of codependency in the second moment of a continuous-time multivariate asset price process, which we name the realized copula of volatility. The statistic is based on local volatility estimates constructed from high-frequency asset returns and affords a nonparametric estimator of the empirical copula of the latent stochastic volatility. We show consistency of our estimator with in-fill asymptotic theory, either with a fixed or increasing time span. In the latter setting, we derive a functional central limit theorem for the empirical process associated with the measurement error of the time-invariant marginal copula of volatility. We also develop a goodness-of-fit test to evaluate hypotheses about the shape of the latter. In a simulation study, we demonstrate that our estimator is a good proxy of both the empirical and marginal copula of volatility, even with a moderate amount of high-frequency data recorded over a relatively short sample. The goodness-of-fit test is found to exhibit size control and excellent power. We implement our framework on high-frequency transaction data from futures contracts that track the U.S. equity and treasury bond market. A Gumbel copula is found to offer a near-perfect bind between the realized variance processes in these data.

\vspace*{0.5cm}

\bigskip \noindent \textbf{JEL Classification}: C14; C32; C58; C80.

\medskip \noindent \textbf{Keywords}: Copula; empirical process; functional central limit theorem; high-frequency data; nonparametric estimation; stochastic volatility; tail dependence.
\end{abstract}

\vfill

\thispagestyle{empty}

\pagebreak

\section{Introduction} \setcounter{page}{1}

Modeling the comovement of stochastic volatility across several assets is paramount to risk measurement and management, portfolio allocation, and multi-asset and derivative pricing. It has long been recognized that financial markets are interconnected, causing asset return volatilities---and correlations---to surge in tandem during financial distress, triggering significant volatility spillover effects, posing a potential systemic risk \citep[see, e.g.,][and references therein]{diebold-yilmaz:09a, billio-getmansky-lo-pelizzon:12a}. An understanding of volatility codependency is also important for pricing of volatility risk (i.e., the variance risk premium) across asset classes \citep{carr-wu:09a, bollerslev-todorov:11a}. The importance of the second moment structure of asset returns is further fueled by the so-called leverage effect that captures a negative association between return and volatility \citep{black:76a, christie:82a}. Thus, from a risk aggregation view the relevant target is often the dependence structure in second moment of returns (namely volatilities and correlations), rather than in the raw returns.

In practice, the assessment of volatility comovements is complicated by the fact that volatility is unobserved, so it requires an empirical route based on an observable proxy. It therefore remains a relatively unexplored area. In recent years, however, the availability of high-frequency data has made it possible to estimate volatility nonparametrically with the realized variance \citep[e.g.,][]{andersen-bollerslev:98a, barndorff-nielsen-shephard:02a}. Most of this literature concentrates on estimation of integrated measures of variation, defined as an integral of a smooth function of the latent spot variance process \citep[see][]{jacod-rosenbaum:13a}. However, as \citet{li-todorov-tauchen:13a} point out, the mapping from the probability distribution of spot volatility to the distribution of integrated volatility is not one-to-one. Thus, the ``integrated'' approach may be partially informative about some distributional aspects of volatility and codependency, but it faces a limitation due to the inherent ``identification failure.'' This may weaken, or outright eliminate, power for testing against relevant alternatives. They suggest to work directly with the pathwise distribution of spot volatility via a so-called occupation measure \citep[e.g.,][]{geman-horowitz:80a}, which recovers the empirical distribution function of the point-in-time volatility process \citep[see also][]{li-todorov-tauchen:16a, li-todorov-tauchen:16b}.\footnote{In addition, \citet{christensen-thyrsgaard-veliyev:19a} and \citet{zhang-li-bollerslev:22a} made the univariate framework robust to microstructure noise.} Regardless, the volatility occupation literature is entirely univariate and has concentrated on the marginal distribution, rather than looking at the interrelationship between various volatility processes.

The starting point of our analysis is the remarkable result from probability theory known as Sklar's theorem \citep[][]{sklar:59a}. It states that for any continuous random vector $(X_{1}, \dots, X_{d})$ with cumulative distribution function $H(x_{1}, \dots, x_{d}) = \mathbb{P}(X_{1} \leq x_{1}, \dots, X_{d} \leq x_{d})$, there exists a unique function $C:[0,1]^{d} \rightarrow [0,1]$ called a \textit{copula}, such that $H(x_{1}, \dots, x_{d}) = C(H_{1}(x_{1}), \dots, H_{d}(x_{d}))$, where $H_{i}(x_{i}) = \mathbb{P}(X_{i} \leq x_{i})$ is the marginal distribution of $X_{i}$.\footnote{The concise mathematical definition is as follows: A $d$-dimensional copula $C:[0,1]^{d} \rightarrow [0,1]$ is the cumulative distribution function of a random vector $(U_{1}, \dots, U_{d})$, where the marginal distribution of $U_{i}$ is uniform, i.e. $C(u_{1}, \dots, u_{d}) = \mathbb{P}(U_{1} \leq u_{1}, \dots, U_{d} \leq u_{d})$ with $\mathbb{P}(U_{i} \leq u_{i}) = u_{i}$, for $i = 1, \dots, d$ and $0 \leq u_{i} \leq 1$.} Conversely, given $C(u_{1}, \dots, u_{d})$ and $H_{1}(x_{1}), \dots, H_{d}(x_{d})$, $H(x_{1}, \dots, x_{d}) = C(H_{1}(x_{1}), \dots, H_{d}(x_{d}))$ is a probability distribution. Copulas provide a marginal-free representation of dependence in a multivariate setting, and they enable the researcher to model the marginal features of the data separately.\footnote{The term copula originates from the Latin word for ``bond'' or ``tie''.} Indeed, many measures of concordance (and discordance), such as Kendall's tau and Spearman's rho, can be expressed entirely in terms of the copula (see Section \ref{section:empirical}). The latter are often more informative than the Pearson's coefficient of correlation, which captures \textit{linear} dependence. As emphasized by \citet{embrechts-mcneil-straumann:99a}, correlation is insufficient for risk management. This is explained by the fact that financial markets typically exhibit \textit{nonlinear} dependence and, further, may possess a high degree of tail concentration (or corner clustering).\footnote{The application of copula theory had a revival in financial economics after \citet{li:00a} proposed to model the default correlation in portfolios of bond loans with a Gaussian copula. This enabled tractable analysis of structured credit products (such as collateralized debt obligations). The model became a standard pricing tool on Wall Street. However, since the Gaussian copula has no tail dependence, it tends to understate the probability of concurrent defaults, as seen during the financial crisis \citep[see, for instance,][for a critique of the approach]{salmon:09a}. A comprehensive review of the application of copula theory to economic and financial time series is \citet{patton:12a}.} Therefore, copulas---or descriptive statistics derived from them---are a better representation of dependence. Another advantage of rank-based statistics is their invariance to monotonic transformations (that do not alter the ordering of the data). However, a model-free technique for analyzing the multivariate dependence structure of volatility is still lacking. In particular, tools for estimating and conducting inference on the copula of volatility are not available.

We close that gap and propose to study the dependence structure of several stochastic volatility processes. We make a four-fold contribution to the literature. First, we extend \citet{li-todorov-tauchen:13a} to the multivariate setting. We operate within a general, albeit standard, continuous-time arbitrage-free It\^{o} semimartingale framework for a (multi-dimensional) log-price process, where the random volatility can exhibit rather unrestricted dynamics. We replace the latent spot variance processes by localized, jump-robust, realized variance measures constructed from discretely observed high-frequency returns over short time intervals. The estimation errors are controllable by the sampling interval and bandwidth selection \citep[e.g.,][]{jacod-protter:12a}. We allow for very active jumps in the price process. The latter can be regarded as ``nuisance parameters'' in our setting, and we handle them with a truncation device \citep[e.g.,][]{mancini:09a}. Second, we construct a nonparametric estimator of volatility codependency, called the \textit{realized copula of volatility}, enabling a meaningful cross-asset comparison irrespective of the marginal distribution of the individual volatility processes. Third, we develop the necessary large-sample theory for our estimator aligned with its target at each horizon. In the in-fill asymptotic limit with fixed time span, we show the realized copula of volatility converges, uniformly, in probability to its empirical counterpart that captures the volatility dependency over the time interval observed so far. Importantly, we allow for nonstationary volatility. In a double asymptotic in-fill and long-span setting---with stationary and weakly dependent volatility---we establish a functional central limit theorem in the form of weak convergence of the measurement error of the realized volatility copula. We show how this can be exploited to construct pointwise or, with a bit more effort, uniform confidence bands. A goodness-of-fit test for assessing the shape of the volatility copula is also proposed. The latter procedure can, for instance, be used to gauge the adequacy of various parametric copulas; an approach we look further at in our application. This builds upon and extends earlier results on the empirical copula process \citep{deheuvels:79a, fermanian-radulovic-wegkamp:04a}. However, as explained above, our framework is distinctly more complicated than the discrete-time setup with independent and identically distributed data (see \citet{fermanian-scaillet:02a} for related work in the time series setting). Fourth, we deliver an implementable inference procedure by designing a \citet{newey-west:94a}-type heteroscedasticity and autocorrelation consistent estimator of the long-run asymptotic variance of our realized copula of volatility.

The small sample properties of our estimator are investigated through a Monte Carlo experiment. We demonstrate its efficacy over the entire support of the copula. In addition, we evaluate the goodness-of-fit test in order to illuminate its size and power. In our empirical work, we consider an application to U.S. equity and treasury futures markets. This leads to forceful evidence that a Gumbel copula---exhibiting upper-tail dependence---offers a near-perfect description of the realized copula of volatility. Thus, our nonparametric procedure complements the existing literature and facilitates a marginal-free analysis of volatility codependency, offering a more refined assessment of systemic risk transmission.

The paper progresses as follows. In Section \ref{section:setting}, we present the theoretical framework and define the empirical copula of the volatility process based on an occupation measure. In Section \ref{section:estimation}, we adopt a spot volatility estimator based on discrete high-frequency data, from which we can construct the realized copula of volatility. In Section \ref{section:theory}, we state our assumptions and develop large-sample theory for our estimator. In particular, we show its consistency for the empirical copula of volatility in the in-fill limit with fixed time span. We also derive a functional central limit theorem for the empirical process associated with the estimation error of the time-invariant copula of volatility in a double asymptotic in-fill and long-span setting. We further propose an estimator of the long-run asymptotic variance and show how to exploit the functional analysis to construct uniform confidence bands. At last, we develop a goodness-of-fit test. Section \ref{section:simulation} reports a detailed Monte Carlo study to evaluate the finite sample properties of our realized copula of volatility, along with a parametric version of our goodness-of-fit test. In Section \ref{section:empirical}, we present an empirical application. We conclude in Section \ref{section:conclusion}. The proofs of our statistical analysis are postponed to the \hyperref[app:proofs]{Appendix}.

\section{The setting} \label{section:setting}

We suppose that a filtered probability space (or stochastic basis) $\big( \Omega, \mathcal{F}, (\mathcal{F}_{t})_{t \geq 0}, \mathbb{P} \big)$ describes the evolution of a pair of log-price processes, $X$ and $Y$, that evolve in continuous-time over an interval $[0,T]$.\footnote{Our theoretical framework extends trivially to $d$-dimensional processes, for any fixed $d \geq 2$. We concentrate on the bivariate setting, because it conveys the main idea and avoids the cost of extra notation.} Here, $\mathcal{F}_{s} \subseteq \mathcal{F}_{t} \subseteq \mathcal{F}$ for $s \leq t \leq T$ is a filtration that represents past and current information available to market participants at any time.

As consistent with the ``no free lunch with vanishing risk'' principle from financial economics, we assume that the motion of $(X,Y)$ is described by an It\^{o} semimartingale \citep{delbaen-schachermayer:94a}.\footnote{The It\^{o} semimartingale assumption is commonly motivated in economics by the efficient markets hypothesis and rational expectations theory \citep[see, e.g.,][]{samuelson:65a}.} This means that it has the component-wise representation
\begin{equation} \label{equation:price-process}
\mathrm{d} Z_{t} = b_{t}^{Z} \mathrm{d}t + \sigma_{t}^{Z} \mathrm{d}W_{t}^{Z} + \int_{| \delta^{Z}(s,z)| \leq 1} \delta^{Z}(t,z)( \mu^{Z}- \nu^{Z})( \mathrm{d}z, \mathrm{d}t) + \int_{| \delta^{Z}(s,z)|> 1} \delta^{Z}(t,z){ \mu}^{Z}( \mathrm{d}z, \mathrm{d}t),
\end{equation}
where the drift $(b_{t}^{Z})_{t \geq 0}$ and the volatility $( \sigma_{t}^{Z})_{t \geq 0}$ are adapted c\`{a}dl\`{a}g processes, while $(W_{t}^{Z})_{t \geq 0}$ is a standard Brownian motion. In addition, $\mu^{Z}$ is a Poisson random measure on $( \mathbb{R}_{+}, \mathbb{R})$ with compensator $\nu^{Z}( \mathrm{d}t, \mathrm{d}z) = dt \otimes \lambda^{Z}( \mathrm{d}z)$, $\lambda^{Z}$ is a $\sigma$-finite measure on $\mathbb{R}$, and $\delta^{Z}: \Omega \times \mathbb{R}_{+} \times \mathbb{R} \rightarrow \mathbb{R}$ is predictable. Here, and in other places, we employ a generic process $Z$ (typically with $Z = X,Y$) to avoid repeating definitions.

In general, a semimartingale can be decomposed into a finite variation component and a local martingale. The ``It\^{o}'' classifier imposes an absolute continuity condition on these, so that we can express them as integrals of ``spot'' processes. This restriction is common, because the latter are more amenable to statistical analysis from high-frequency data. Apart from that, our framework is nonparametric, model-free, and suffices to capture the most dominant features of empirical asset price processes, such as time-varying expected returns, stochastic volatility, and leverage effect. Moreover, it encompasses the presence of large and small price jumps, where the latter can be infinitely active and of infinite variation. We also permit complex within- and cross-asset dependencies between the various model components. For example, the standard Brownian motions can be correlated (i.e. $\rho_{t} \mathrm{d}t = \mathrm{d} \langle W^{X}, W^{Y} \rangle_{t}$, where $\langle W^{X}, W^{Y} \rangle_{t}$ is the quadratic covariation process). Moreover, the jumps in volatility can be related to the occurrence of jumps in the log-price, and there can be common jumps in both \citep[for instance][]{todorov-tauchen:11a, bibinger-winkelmann:18a}. We are, however, going to impose additional structure and smoothness conditions for our econometric procedure.

Our aim is to study the empirical copula of volatility, which begins with a notion of the empirical distribution function of volatility, defined as
\begin{equation} \label{equation:edf}
H_{T}(x,y) = \frac{1}{T} \int_{0}^{T} \mathbbm{1}_{ \left \{V_{t}^{X} \leq x, V_{t}^{Y} \leq y \right \}} \mathrm{d}t,
\end{equation}
where $0 < x,y < \infty$, while $V_{t}^{X} = ( \sigma_{t}^{X})^{2}$ and $V_{t}^{Y} = ( \sigma_{t}^Y)^{2}$ are point-in-time variances of $X$ and $Y$, and $\mathbbm{1}_{ \left\{ \cdot \right\}}$ is the indicator function. $H_{T}(x,y)$ measures the relative amount of time in $[0,T]$ that the stochastic volatility processes spent at different levels over their support. Hence, \eqref{equation:edf} is the pathwise counterpart of the cumulative distribution function of volatility. Without the $T$ normalization, $H_{T}(x,y)$ is a multivariate extension of the volatility occupation measure studied in \citet{li-todorov-tauchen:13a}, which corresponds to the marginal empirical distribution function $F_{T}(x) = H_{T}(x, \infty)$ and $G_{T}(y) = H_{T}( \infty, y)$.

The empirical copula of volatility is motivated by Sklar's theorem. In particular, since $H_{T}(x,y)$ in \eqref{equation:edf} is the pathwise analogue of the distribution function of the latent variance process $(V_{t}^{X}, V_{t}^{Y})$ over $[0, T]$, it is natural to separate its marginal features from its dependence structure. We therefore define the empirical copula of volatility as
\begin{equation} \label{equation:empirical-copula}
C_{T}(u,v) = \frac{1}{T} \int_{0}^{T} \mathbbm{1}_{ \left \{F_{T}(V_{t}^{X}) \leq u, G_{T}(V_{t}^{Y}) \leq v \right \}} \mathrm{d}t,
\end{equation}
where $0 < u, v < 1$.

If $F_{T}(x)$ and $G_{T}(y)$ are invertible (e.g., continuous and strictly increasing), \eqref{equation:empirical-copula} can be written as $C_{T}(u,v) = H_{T}(F_{T}^{-1}(u),G_{T}^{-1}(v))$, which leads to the ``inversion formula'' $H_{T}(x,y) = C_{T}(F_{T}(x),G_{T}(y))$. Thus, the empirical distribution function of volatility can be expressed as the marginal empirical distribution of the individual volatility processes and a copula that captures their dependence structure up to $T$. $C_{T}(u,v)$ delivers a margin-free characterization of volatility codependency that is invariant to strictly monotone transformations of the marginal volatility processes. This is useful in financial applications, where the marginal behavior of volatility can differ markedly across assets, while their dependence structure can be decided separately. In addition, and in contrast to Pearson's measure of linear correlation, the copula can accommodate nonlinear dependence and tail concentration.

Without prior knowledge of the volatility processes, however, we define the empirical quantile function of $F_{T}(x)$ and $G_{T}(y)$ as the generalized inverse:
\begin{equation} \label{equation:quantile-function}
Q_{T}^{X}(u) \equiv F_{T}^{-1}(u) = \inf_{x \geq 0} \{x: F_{T}(x) \geq u \} \quad \text{and} \quad Q_{T}^{Y} (v) \equiv G_{T}^{-1}(v) = \inf_{y \geq 0} \{y: G_{T}(y) \geq v \}
\end{equation}
with the convention $\inf \emptyset = \infty$.\footnote{The generalized inverse of a distribution function is non-decreasing, left-continuous, and admits a limit from the right. If the distribution function is discontinuous, the quantile function has flat spots, and vice versa.} In the general setting, we still set $C_{T}(u,v) = H_{T}(Q_{T}^{X}(u),Q_{T}^{Y}(v))$, but it should now be given an infimum-based interpretation. Hence, $C_{T}(u,v)$ measures the proportion of time that the volatility processes spend below the smallest coordinate $x$ and $y$ for which $F_{T}(x)$ and $G_{T}(y)$ exceed $u$ and $v$.

To the best of our knowledge, both the multivariate extension of the empirical distribution function of volatility, and the associated copula, are novel to the literature.

\section{High-frequency estimation} \label{section:estimation}

In practice, $V^{X}$ and $V^{Y}$ are not directly observed, so neither is the empirical distribution function of volatility nor the copula transformation. Our empirical strategy is based on the premise that we can approximate the spot variance processes with a realized measure constructed from discretely observed high-frequency data of $X$ and $Y$ over small time intervals. In the asymptotic theory, we are going to both increase the number of blocks and shrink the timespan of each block, while padding it with an increasing number of observations, such that we attain a better and better assessment of the volatility's path. In doing so, we control for the price jump component. We also deal with the drift, but this is easier since it is asymptotically negligible in our setting.\footnote{\citet{christensen-oomen-reno:22a} propose a drift burst model, in which the drift term is locally of larger order than the volatility component.}

We assume that $(X,Y)$ are sampled at equidistant time points $(i \Delta_{n})_{i=0}^{n}$, where $\Delta_{n}$ is the time gap and $n = \lfloor T/ \Delta_{n} \rfloor$ denotes the number of observations over $[0, T]$. We define an increment of the process $Z$ between $(i-1) \Delta_{n}$ and $i \Delta_{n}$ as $\Delta_{i}^{n}Z = Z_{i \Delta_{n}} - Z_{(i-1) \Delta_{n}}$, for $i = 1, \dots, n$. Then, as in \citet{li-todorov-tauchen:13a}, we adopt the following estimator of $V_{t}^{Z}$ (at time $t$):
\begin{equation} \label{equation:realized-variance}
\begin{aligned}
\hat{V}_{t}^{Z} = \left\{
\begin{array}{ll}
\frac{1}{k_{n} \Delta_{n}} \sum_{j=1}^{k_{n}}\big( \Delta_{k_{n} \lfloor \frac{t}{k_{n} \Delta_{n}} \rfloor+j}^{n} Z \big)^{2} \mathbbm{1}_{ \{| \Delta_{k_{n} \lfloor \frac{t}{k_{n} \Delta_{n}} \rfloor+j}^{n} Z| \leq \alpha \Delta_{n}^{ \varpi} \}}, & 0 \leq t< \lfloor \frac{T}{k_{n} \Delta_{n}} \rfloor k_{n} \Delta_{n}, \\[0.25cm]
\frac{1}{k_{n}\Delta_{n}} \sum_{j=1}^{k_{n}} \left( \Delta_{n-j+1}^{n} Z \right)^{2} \mathbbm{1}_{ \left \{ \left| \Delta_{n-j+1}^{n} Z \right| \leq \alpha \Delta_{n}^{ \varpi} \right \}}, & \lfloor \frac{T}{k_{n} \Delta_{n}} \rfloor k_{n} \Delta_{n} \leq t \leq T.
\end{array} \right.
\end{aligned}
\end{equation}
$\hat{V}_{t}^{Z}$ is the spot realized variance \citep[e.g.,][]{jacod-protter:12a}. $k_{n}$ is the number of log-price increments included in the estimation on each subinterval and tends to infinity in the asymptotic analysis, whereas $\alpha > 0$ and $\varpi \in (0, 1/2)$ relate to the jump truncation device of \citet{mancini:09a}. The threshold $\alpha \Delta_{n}^{ \varpi}$ is designed to remove log-price increments perturbed by the price jump component and is asymptotically shrinking, slowly enough, to zero. The rate conditions of the tuning parameters are made explicit below.\footnote{Note that in much of the following, we suppress the explicit dependence of various statistics on $n$, when it can be avoided without causing confusion.}

We define the realized distribution function of volatility, an estimator of the empirical distribution function of volatility, as follows:
\begin{equation}
\widehat{H}_{n,T}(x,y) = \frac{1}{T} \int_{0}^{T} \mathbbm{1}_{ \left \{ \hat{V}_{t}^{X} \leq x, \hat{V}_{t}^{Y} \leq y \right \}} \mathrm{d}t.
\end{equation}
The univariate version can be retrieved as $\widehat{F}_{n,T}(x) = \widehat{H}_{n,T}(x, \infty) = \frac{1}{T} \int_{0}^{T} \mathbbm{1}_{ \left \{ \hat{V}_{t}^{X} \leq x \right\}} \mathrm{d}t$ and $\widehat{G}_{n,T}(y) = \widehat{H}_{n,T}(\infty,y) = \frac{1}{T} \int_{0}^{T} \mathbbm{1}_{ \left\{ \hat{V}_{t}^{Y} \leq y \right\}} \mathrm{d}t$. The latter are again equivalent to those proposed in \citet{li-todorov-tauchen:13a}.

Then, the realized copula of volatility is the statistic:
\begin{equation} \label{equation:realized-copula}
\quad \widehat{C}_{n,T}(u,v) = \frac{1}{T} \int_{0}^{T} \mathbbm{1}_{ \left\{ \widehat{F}_{n,T}( \hat{V}_{t}^{X}) \leq u, \widehat{G}_{n,T}( \hat{V}_{t}^{Y}) \leq v \right \}} \mathrm{d}t,
\end{equation}
and the realized quantile functions of volatility are extended as follows:
\begin{equation}
\widehat{Q}_{n,T}^{X}(u) = \inf_{x \geq 0} \{x: \widehat{F}_{n,T}(x) \geq u\} \qquad \text{and} \qquad \widehat{Q}_{n,T}^{Y}(v) = \inf_{y \geq 0} \{y: \widehat{G}_{n,T}(y) \geq v \},
\end{equation}
such that $\widehat{C}_{n,T}(u,v) = \widehat{H}_{n,T}( \widehat{Q}_{n,T}^{X}(u), \widehat{Q}_{n,T}^{Y}(v))$.

\section{Asymptotic theory} \label{section:theory}

To derive the asymptotic theory for the realized distribution function of volatility and realized copula of volatility, we prove a consistency result in distinct settings with i) $\Delta_{n} \rightarrow 0$ and $T$ fixed, and ii) $\Delta_{n} \rightarrow 0$ and $T \rightarrow \infty$. In the first, the realized statistic is a natural estimator of the empirical counterpart. In the second, the realized statistic converges to the invariant distribution, assuming it exists. Under further smoothness conditions and weak dependence of the volatility process, we establish a central limit theorem. The latter is done in the functional sense of weak convergence of the probability measures associated with the empirical process of the normalized measurement error of volatility.

To begin with, we introduce a number of additional conditions. The first assumption concerns the regularity of the price jump component.

\begin{asu} \label{assumption:jump-activity}
The process $Z$ follows \eqref{equation:price-process} with $b^{Z}$ locally bounded and $\sigma^{Z}$ c\`{a}dl\`{a}g. Moreover, for a constant $r \in [0,2]$, and a sequence of stopping times $( \tau_{m})_{m=1}^{ \infty}$, such that $\tau_{m} \rightarrow \infty$ as $m \rightarrow \infty$, there exists a sequence of functions $( \Gamma_{m})_{m=1}^{ \infty}$ such that for each $m: \min(| \delta^{Z}( \omega,t,z)|,1) \leq \Gamma_{m}(z)$ for $( \omega,t,z)$ with $t \leq \tau_{m}( \omega)$, and $\int_{ \mathbb{R}} \Gamma_{m}(z)^{r} \lambda^{Z}( \mathrm{d}z) < \infty$.
\end{asu}
The exponent $r$ in Assumption \ref{assumption:jump-activity} is an upper bound on the generalized Blumenthal-Getoor index, an adaptation from L\'{e}vy processes to the semimartingale setting, see \citet[][Lemma 3.2.1]{jacod-protter:12a}.\footnote{The Blumenthal-Getoor index of a L\'{e}vy process is defined as $\beta = \inf \left\{ r \geq 0 \;:\; \int_{|x| < 1} |x|^{r} \nu( \mathrm{d}x) < \infty \right\}$, where $\nu$ is the L\'{e}vy measure. $\beta$ can be interpreted as the smallest power at which the distribution of the small jumps (arbitrarily chosen to be of size $|x| < 1$) has finite $r$th moment.} It is therefore related to the jump activity of the log-price process with lower values of $r$ being more binding. The point $r = 1$ separates jump processes with sample paths of finite and infinite variation, i.e., it concerns their absolute summability, while $r = 0$ implies that the jumps are finitely active (on finite time intervals, almost surely). $r$ plays a crucial role in determining rate conditions for the tuning parameters in our estimation procedure, and also in establishing the rate of convergence of our estimator.

\begin{asu} \label{asucont}
$H_{T}(x,y)$ is continuous, almost surely.
\end{asu}

The continuity of the empirical distribution function imposed by Assumption \ref{asucont} is used to control the approximation error in the uniform metric.

\begin{thm} \label{theorem:consistency}
We suppose that Assumption \ref{assumption:jump-activity} with $r = 2$ and Assumption \ref{asucont} hold. As $\Delta_{n} \rightarrow 0$ and $k_{n} \rightarrow \infty$, such that $k_{n} \Delta_{n} \rightarrow 0$, and with $T$ fixed, it follows that
\begin{equation}
\sup_{(x,y) \in \mathbb{R}_{+}^{2}}| \widehat{H}_{n,T}(x,y) - H_{T}(x,y)| \overset{ \mathbb{P}}{ \longrightarrow} 0.
\end{equation}
Furthermore, let $\mathcal{Q} = \{(u,v) \in [0,1]^{2} : Q_{T}^{X}(u)$ is continuous at $u$ and $Q_{T}^{Y}(v)$ is continuous at $v$, almost surely$\}$. Then, for each $(u,v) \in \mathcal{Q}$, it further holds that
\begin{equation}
\widehat{C}_{n,T}(u,v) \overset{ \mathbb{P}}{ \longrightarrow} C_{T}(u,v).
\end{equation}
Moreover, if $Q_{T}^{X}(u)$ and $Q_{T}^{Y}(v)$ are continuous, almost surely, on $[0,1]^{2}$:
\begin{equation}
\sup_{(u,v) \in [0,1]^{2}}| \widehat{C}_{n,T}(u,v) - C_{T}(u,v)| \overset{ \mathbb{P}}{ \longrightarrow} 0.
\end{equation}
\end{thm}

Theorem 1 establishes uniform convergence in probability of $\widehat{H}_{n,T}(x,y)$ as an estimator of $H_{T}(x,y)$. As $\Delta_{n} \rightarrow 0$ and $k_{n} \rightarrow \infty$, such that $\Delta_{n} k_{n} \rightarrow 0$, we get an increasing number of error-free estimates of the spot variance, forming a dense subset on $[0,T]$, from which we can construct an arbitrarily accurate estimate of $H_{T}(x,y)$. The exact rate of $k_{n}$ does not influence this result, so long as $\Delta_{n} k_{n} \rightarrow 0$. We can also show pointwise convergence in probability of $\hat{C}_{n,T}(x,y)$ on the continuity set $\mathcal{Q}$. However, in the discontinuous case at a point $(u,v)$ where either $Q_{T}^{X}(u)$ or $Q_{T}^{Y}(v)$ are not continuous, even uniform convergence of $\widehat{H}_{n,T}(x,y)$ does not imply that $Q_{n,T}^{X}(u) \overset{ \mathbb{P}}{ \longrightarrow} Q_{T}^X(u)$ (or $Q_{n,T}^Y(u) \overset{ \mathbb{P}}{ \longrightarrow} Q_{T}^{Y}(v)$), so $\widehat{C}_{n,T}(u,v)$ is generally inconsistent for $\widehat{C}_{T}(u,v)$ at that point. We can strengthen the convergence of $\widehat{C}_{n,T}(x,y)$ to a uniform statement if the limiting empirical quantile function is continuous, almost surely.

The smoothness imposed on $H_{T}(x,y)$ does not require the volatility path itself to be continuous. Even if $V_{t}^{X}$ or $V_{t}^{Y}$ exhibit jumps, $H_{T}(x,y)$ can be continuous provided the sample path variation is rich enough to fill out the gaps. The intuition is that $H_{T}(x,y)$ aggregates the amount of time that the processes spend below a given threshold, and it can thus ``iron out'' local irregularities through temporal aggregation.

To establish the asymptotic theory for the combined in-fill and long-span setting, we need stricter control of the error embedded in the recovery of the volatility path. First, we need to restrict the c\`{a}dl\`{a}g dynamic of the volatility processes.

\begin{asu} \label{assumption:sigma}
The volatility process $\sigma^{Z}$ is of the form:
\begin{equation} \label{equation:sigma}
\sigma_{t}^{Z} = \sigma_{0}^{Z} + \int_{0}^{t} \tilde{b}_{s}^{Z} \mathrm{d}s + \int_{0}^{t} \tilde{ \sigma}_{s}^{Z} \mathrm{d}\tilde{W}_{s}^{Z} + \int_{0}^{t} \int_{ \mathbb{R}} \tilde{ \delta}^{Z}(s,z)( \tilde{ \mu}^{Z} - \tilde{ \nu}^{Z})(  \mathrm{d}s,  \mathrm{d}z).
\end{equation}
Here, $\tilde{b}^{Z}$ and $\tilde{ \sigma}^{Z}$ are adapted and locally bounded processes, $\tilde{W}^{Z}$ is a standard Brownian motion (that may depend on $W^{Z}$), while $\tilde{ \delta}^{Z}$ is a predictable function. Moreover, for a constant $\tilde{r} \in [0,2]$, and a sequence of stopping times $(\tilde{ \tau}_{m})_{m=1}^{ \infty}$ such that $\tilde{ \tau}_{m} \rightarrow \infty$ as $m \rightarrow \infty$, there exists a sequence of functions $( \tilde{ \Gamma}_{m})_{m=1}^{ \infty} : \min \{| \tilde{ \delta}^{Z}( \omega,t,z)|,1 \} \leq \tilde{ \Gamma}_{m}(z)$ for $( \omega,t,z)$ with $t \leq \tilde{ \tau}_{m}( \omega)$, and $\int_{ \mathbb{R}} \tilde{ \Gamma}_{m}(z)^{ \tilde{r}} \tilde{ \lambda}^{Z}( \mathrm{d}z) < \infty$.
\end{asu}

Assumption \ref{assumption:sigma} states that the volatility is an It\^{o} semimartingale. In contrast to the log-price process, this is not a consequence of the no-arbitrage principle. However, it is commonly assumed, because it restricts the local behavior of volatility that is necessary to apply standard estimates for semimartingales in the high-frequency paradigm \citep[see][]{jacod-protter:12a}. For example, if $\sigma^{Z}$ is a diffusion, the assumption follows by an application of It\^{o}'s Lemma under certain smoothness conditions. Either way, it is fulfilled for most stochastic volatility models applied in practice, even when volatility has sources of randomness not captured by $Z$. The ability to correlate the Brownian motions of $Z$ and $\sigma^{Z}$ is particularly relevant for financial applications that involve equity data, in order to capture the so-called leverage effect \citep[][]{black:76a, christie:82a}. Note that \eqref{equation:sigma} does rule out volatility models driven by a fractional Brownian motion, which feature prominently in some strands of research \citep[e.g.,][]{comte-renault:98a, gatheral-jaisson-rosenbaum:18a}. The main difficulty is that in the ``rough'' regime, the volatility behaves very erratically over short time intervals, and this makes it harder to control the discretization error \citep[see][]{chong-todorov:23a}. \citet{christensen-thyrsgaard-veliyev:19a} allow $\sigma^{Z}$ to satisfy a H\"{o}lder-type condition, in the expectation sense. It may be possible to extend our theoretical framework in this direction as well, but we leave it for a future endeavor.

Second, we restrict the memory of volatility processes. To accomplish this, we introduce the additional notation for a generic process $Z$:
\begin{equation} \label{equation:mixing}
\alpha_{t} = \sup_{A \in \mathcal{F}_{s}, B \in \mathcal{F}^{s+t}}| \mathbb{P}(A \cap B) - \mathbb{P}(A) \mathbb{P}(B)|. \\
\end{equation}
Here, $\mathcal{F}_{s} = \sigma(Z_{u}; u \leq s)$ and $\mathcal{F}^{t} = \sigma(Z_{u}; u \geq t)$ are the backward- and forward-looking $\sigma$-algebras generated by $Z$ until time $s$ and from time $t$, respectively. Moreover, $\alpha_{t}$ is the mixing coefficient function that measures the degree of stochastic (in)dependence between them. This is used in the definition of strong mixing introduced by \cite{rosenblatt:56a}.

\begin{definition} \label{definition:strong-mixing}
$Z$ is strong mixing (or $\alpha$-mixing) if $\alpha_{t} \rightarrow 0$ as $t \rightarrow \infty$.
\end{definition}

\begin{asu} \label{asu5}
The volatility process $(V_{t}^{X}, V_{t}^{Y})_{t \geq 0}$ is stationary and $\alpha$-mixing in the sense of Definition \ref{definition:strong-mixing} with $\alpha_{t} = O(t^{-(1+ \tau)})$ for some $\tau > 0$.
\end{asu}

Assumption \ref{asu5} is essential in deriving the asymptotic distribution theory. It corresponds to Assumption 5 in \cite{christensen-thyrsgaard-veliyev:19a}, but for the two-dimensional case, and is again a weak regularity condition that can be verified for most stochastic volatility models. The stationary condition makes it meaningful to speak of a time-invariant distribution function of volatility, $H(x,y) = \mathbb{P}(V_{0}^{X} \leq x, V_{0}^{Y} \leq y)$, from which $F(x) = \mathbb{P}(V_{0}^{X} \leq x)$ and $G(y) = \mathbb{P}(V_{0}^{Y} \leq y)$ can be extracted. Sklar's theorem then guarantees existence, and uniqueness under further conditions, of the copula $C(u,v) = \mathbb{P}(F(V_{0}^{X}) \leq u, G(V_{0}^{Y}) \leq v)$. The question is then to what extent $\widehat{H}_{n,T}(x,y)$, consistently estimating $H_{T}(x,y)$ as $\Delta_{n} \rightarrow 0$, also provides a good proxy of the stationary distribution $H(x,y)$, since there is only a single realization to work with. This, at a bare minimum, additionally requires $T \rightarrow \infty$, so sufficient conditions to apply the ergodic theorem on the volatility path are needed to prompt the law of large numbers for consistency. This is implied by the strong mixing condition, and the polynomial decay of $\alpha_{t}$ ensures that volatility persistence decays suitably fast to derive the limiting distribution of the statistic. In this regard, we should note that since our CLT is derived for a sequence of bounded random variables, all moments exist, so the usual moment condition ``trade-off'' in Davydov’s inequality is irrelevant, and hence the requirement for the mixing coefficient function is the weakest possible and the autocorrelation function merely needs to be integrable.

A prominent counterexample (ruled out by Assumption \ref{assumption:sigma}), where Assumption \ref{asu5} typically fails, is long-memory processes, such as the fractional Gaussian noise with Hurst exponent $H \in (0.5,1.0)$, which is stationary and weak mixing, but not strong mixing.

The next assumption imposes a uniform boundedness on the coefficient processes.

\begin{asu} \label{assumption:moment}
For all $p \geq 1$,
\begin{equation}
\sup_{t \geq 0} \left\{ \mathbb{E} \left[|b_{t}^{Z}|^{p} + | \sigma_{t}^{Z}|^{p} + | \tilde{b}_{t}^{Z}|^{p} + | \tilde{ \sigma}_{t}^{Z}|^{p} \right] + \int_{ \mathbb{R}}(1 \wedge \Gamma(z))^{p} \lambda^{Z}( \mathrm{d}z) \right \} \leq C,
\end{equation}
for a constant $C$.
\end{asu}
This follows Assumption 1 in \citet{christensen-thyrsgaard-veliyev:19a} and \citet{andersen-thyrsgaard-todorov:19a}. It imposes a light-tailedness condition on the distribution of the various coefficient processes driving the continuous part of $X$ and $Y$, and the large price jumps, but not for the jumps of volatility. It is a common assumption in the long-span setting. The existence of moments of any order is more than we require, but it streamlines the exposition. The technical reason for this condition is that the standard estimates for the increments of semimartingales based on the localization procedure of \citet[][Section 4.4.1]{jacod-protter:12a} presumes $T$ is fixed. The above assumption facilitates the derivation of such estimates over $[0, \infty)$.

\begin{asu} \label{asu4}
$H_{T}(x,y)$ is continuously differentiable on $\mathbb{R}_{+}^{2}$, almost surely. The partial derivatives (that exist, almost surely) are denoted by $\partial_{x} H_{T}(x,y)$ and $\partial_{y} H_{T}(x,y)$, respectively, and we assume that $\partial_{ \bullet}H_{T}(x,y) \neq 0$ such that $\mathbb{E} \left[ \partial_{ \bullet} H_{T}(x,y) \right] \leq C$ for almost every $(x,y) \in \mathbb{R}_{+}^{2}$, where $C$ is a bounding constant that does not depend on $(x,y)$.
\end{asu}
This combines Assumption 2 of \citet{christensen-thyrsgaard-veliyev:19a} with Assumption B of \citet{li-todorov-tauchen:13a}. The smoothness condition is used to derive a uniform upper bound for estimation error of the empirical copula function. The requirement $\partial_{ \bullet}H_{T}(x,y) \neq 0$ rules out ``flat spots'' in the marginal distribution to ensure that $H_{T}(x,y) \mapsto {C}_T(u,v)$ is smooth and non-singular, so that we can apply the functional delta method without getting a degenerate limit process with zero variance. We note that the boundedness in expectation is weaker than asking $\mathbb{E}[ \sup_{(x,y) \in \mathbb{R}_{+}^{2}} \partial_{ \bullet}H_{T}(x,y)] \leq C$, since the latter requires almost sure control of the pathwise behavior of volatility.

At last, we state an extra condition that is required to establish the central limit theorem. To accomplish this, we introduce the notation $k_{n} = \Delta_{n}^{- \gamma}$ and for any $\iota >0$
\begin{eqnarray} \label{equation:weird-rate-condition}
d_{n}&=&\Delta_{n}^{\gamma-1+(2-r)\varpi}+\Delta_{n}^{\gamma/2-\iota}+\Delta_{n}^{(1-\gamma)/2-\iota}+\Delta_{n}^{\frac{1-\gamma}{1+\tilde{r}}-\iota} \text{ and } \nonumber\\
d_{n}'&=&\Delta_{n}^{\gamma/r-(1-\varpi)-\iota}+\Delta_{n}^{\gamma/2-\iota}+\Delta_{n}^{(1-\gamma)/2-\iota}+\Delta_{n}^{\frac{1-\gamma}{1+\tilde{r}}-\iota}.
\end{eqnarray}
\begin{asu} \label{asu7}
We assume that $\frac{r-1}{r} < \varpi < \frac{1}{2}$ and $r(1- \varpi) < \gamma <1$ for $r>1$. In addition, $0 < \varpi < \frac{1}{2}$ and $1-(2-r) \varpi < \gamma <1$ for $r \leq 1$. Moreover, $\frac{1- \gamma}{1+ \tilde{r}} > 0$ for $\tilde{r}>0$. Finally, we assume that either of the following conditions holds for any $\iota > 0$:
\begin{equation}
1). \:\: r\leq 1 \text{ and } \sqrt{T}d_{n} \rightarrow 0 \quad \text{or} \quad 2). \:\: r > 1 \text{ and } \sqrt{T}d'_{n} \rightarrow 0.
\end{equation}
\end{asu}
This is identical to the conditions in Theorem 4 of \citet{li-todorov-tauchen:13a}, except for the appearance of $\sqrt{T}$ since we operate in the long-span setting. It implies that the volatility error, $\hat{V}_{t}^{Z} - V_{t}^{Z}$, is dominated by $T$, so the high-frequency discretization part is asymptotically negligible. The four terms in $d_n$ and $d_n'$ represent the various sources of the estimation error in $H_{n,T}(x,y)$. Specifically, the first term captures the order of the truncation procedure, induced by replacing truncated increments by those of the continuous part. It depends on whether the jump component has finite or infinite variation. The second and third terms capture the sampling variability and the discretization bias in approximating the spot volatilities. The last order is for the collection of summands over the sampling intervals that contain jumps (with a size larger than some level) in the volatility. The choices of $\varpi, r$, and $\gamma$ imply $d_{n} \rightarrow 0$ and $d_{n}' \rightarrow 0$. Thus, we can assume $T \rightarrow \infty$ such that Assumption \ref{asu7} holds. In particular, for $r \in(0, \frac{1}{2})$ and $\tilde{r}=0$, we can take $\varpi \in( \frac{3}{8-4r}, \frac{1}{2})$ and $\gamma=1/2$, from which it follows that $d_{n}=O( \Delta_n^{ \frac{1}{4}- \iota})$. In this case, $T = O( \Delta_{n}^{-1/4+ \epsilon})$ with $\epsilon > 0$ verifies the requirement.

The next theorem delivers the functional CLT of the estimation procedure for the stationary distribution function of volatility and the associated copula. It is grounded in a double asymptotic setting with $\Delta_{n} \rightarrow 0$ and $T \rightarrow \infty$ and extends Theorem 1 in \citet{li-todorov-tauchen:13a} and Theorem 3.1 in \citet{christensen-thyrsgaard-veliyev:19a}.

\begin{thm} \label{thm2}
Suppose that Assumptions \ref{assumption:jump-activity} - \ref{asu7} hold (with $r = 2$ in Assumption \ref{assumption:jump-activity} and $\tilde{r} = 2$ in Assumption \ref{assumption:sigma}). As $\Delta_{n} \rightarrow 0$ and $T \rightarrow \infty$, it follows that
\begin{equation}
\sqrt{T} \left( \widehat{H}_{n,T}(x,y) - H(x,y) \right) \Rightarrow \mathcal{G}.
\end{equation}
Here, ``$\Rightarrow$'' denotes weak convergence in the space $\mathbb{D}( \mathbb{R}_{+}^{2})$ of c\`{a}dl\`{a}g functions equipped with the uniform topology. Also, $\mathcal{G}$ is a Brownian bridge on $\mathbb{R}_{+}^{2}$ with covariance function
\begin{equation} \label{equation:covariance-matrix-G}
\mathrm{cov} \big( \mathcal{G}(x,y), \mathcal{G}(x',y') \big) \equiv \mathrm{avar}_{ \mathcal{G}}(x,y,x',y') = 2 \int_{0}^{ \infty} \Big( H_{t}(x,y,x',y') - H(x,y)H(x',y') \Big) \mathrm{d}t,
\end{equation}
where $H_{t}(x,y,x',y') = P(V_{0}^{X} \leq x, V_{0}^{Y} \leq y, V_{t}^{X} \leq x', V_{t}^{Y} \leq y')$. Moreover, it follows that
\begin{equation}
\sqrt{T} \left( \widehat{C}_{n,T}(u,v) - C(u,v) \right) \Rightarrow \mathcal{C},
\end{equation}
where $\mathcal{C}$ is a stochastic process on $[0,1]^{2}$, derived from $\mathcal{G}$, that is defined as:
\begin{equation} \label{definition:C}
\mathcal{C}(u,v) = \mathbf{g}(u,v)^{ \top} \mathbf{Z}(u,v).
\end{equation}
Here,
\begin{equation}
\begin{aligned}
\mathbf{Z}(u,v) &= \left[ \mathcal{G} \big(Q^{X}(u),Q^{Y}(v) \big), \mathcal{G} \big(Q^{X}(u), \infty \big), \mathcal{G} \big( \infty,Q^{Y}(v) \big) \right]^{ \top}, \\[0.25cm]
\mathbf{g}(u,v) &= \left[1, -\partial_{u}C(u,v), -\partial_{v}C(u,v) \right]^{ \top},
\end{aligned}
\end{equation}
while $Q^{X}(u) = \lim_{T \rightarrow \infty}Q_{T}^{X}(u)$ and $Q^{Y}(v) = \lim_{T \rightarrow \infty}Q_{T}^{Y}(v)$.
\end{thm}

The proof of Theorem \ref{thm2} draws on the functional CLT for stationary mixing sequences of bounded random variables (see Theorem 18.5.4 in \cite{ibragimov:75a} and Theorem 5.2 in \citet{dehay:05a}) building on \citet{rosenblatt:56a}. We reiterate that the rate of convergence and the asymptotic distribution of $\widehat{H}_{n,T}(x,y)$ and $\widehat{C}_{n,T}(x,y)$ are unaffected by the error from recovering volatility, $V_{t}^{Z} - \widehat{V}_{t}^{Z}$, which is asymptotically negligible under the rate conditions in \eqref{equation:weird-rate-condition}.

The covariance function of $\mathcal{C}(u,v)$ follows from Theorem \ref{thm2} and the functional delta rule, i.e.
\begin{equation} \label{equation:avarC}
\text{avar}_{ \mathcal{C}}(u,v,u',v') \equiv \mathrm{cov} \big( \mathcal{C}(u,v), \mathcal{C}(u',v') \big) = \mathbf{g}(u,v)^{ \top} \Gamma(u,v,u',v') \mathbf{g}(u',v'),
\end{equation}
where
\begin{equation} \label{equation:Gamma}
\Gamma(u,v,u',v') = \text{cov} \big( \mathbf{Z}(u,v), \mathbf{Z}(u',v') \big).
\end{equation}
As a consequence,
\begin{equation} \label{equation:cid}
\begin{aligned}
\sqrt{T} \left( \widehat{H}_{n,T}(x,y) - H(x,y) \right) &\overset{d}{ \longrightarrow} N \big(0, \text{avar}_{ \mathcal{G}}(x,y,x,y) \big), \\
\sqrt{T} \left( \widehat{C}_{n,T}(u,v) - C(u,v) \right) &\overset{d}{ \longrightarrow} N \big(0, \text{avar}_{ \mathcal{C}}(u,v,u,v) \big),
\end{aligned}
\end{equation}
where $\overset{d}{ \longrightarrow}$ is regular convergence in law. This result can be used to construct pointwise confidence intervals, or conduct hypothesis tests, once an estimator of the asymptotic covariance function has been designed.

To construct such an estimator, we set:
\begin{equation} \label{escov_G}
\widehat{ \text{avar}}_{ \mathcal{G}}(x,y,x',y') = 2 \int_{0}^{T^{ \xi}} \left\{ \widehat{H}_{t,n,T}(x,y,x',y') - \widehat{H}_{n,T}(x,y) \widehat{H}_{n,T}(x',y') \right\} \mathrm{d}t,
\end{equation}
with $\widehat{H}_{t,n,T}(x,y,x',y') = \frac{1}{T-T^{ \xi}} \int_{0}^{T-T^{ \xi}} \mathbbm{1}_{ \left\{ \hat{V}_{s}^{X} \leq x, \hat{V}_{s}^{Y} \leq y, \hat{V}_{s+t}^{X} \leq x', \hat{V}_{s+t}^{Y} \leq y' \right\}} \mathrm{d}s$, where $\xi$ is a small positive constant.

\begin{corollary} \label{cor1}
We suppose that Assumptions \ref{assumption:jump-activity} - \ref{asu7} hold, and that $\xi \in (0, 1/3)$. As $\Delta_{n} \rightarrow 0$ and $T \rightarrow \infty$, we deduce that
\begin{equation}
\widehat{ \mathrm{avar}}_{ \mathcal{G}}(x,y,x',y') \overset{ \mathbb{P}}{ \longrightarrow} \mathrm{avar}_{ \mathcal{G}}(x,y,x',y').
\end{equation}
\end{corollary}
In the above, $T^{ \xi}$ plays the role of the ``lag length'' in the estimation of the long-run variance of a stationary time series, which should not grow too fast for a consistent estimation. In fact, the optimal order of the bandwidth is often $T^{1/3}$ \citep[see, e.g.,][]{newey-west:94a}. The choice of $\xi$ relates to this observation.

By Slutsky's theorem,
\begin{equation}
\frac{ \sqrt{T} \left( \widehat{H}_{n,T}(x,y) - H(x,y) \right)}{ \sqrt{ \widehat{ \mathrm{avar}}_{ \mathcal{G}}(x,y,x,y)}} \overset{d}{ \longrightarrow} N(0,1).
\end{equation}
Conducting inference about $C(u,v)$ is a more complicated endeavor, since it also involves partial derivatives of $C(u,v)$ on top of the covariance function $\mathrm{avar}_{ \mathcal{G}}(x,y,x',y')$. To construct estimators of the former, we propose a nonparametric kernel-based smoothing approach:\footnote{This idea was introduced in \cite{rosenblatt:56b}.}
\begin{equation} \label{estparder}
\begin{aligned}
\widehat{ \partial_{u}C}(u,v) &= \frac{1}{hT} \int_{0}^{T} \int_{0}^{v}K \left( \frac{\widehat{F}_{n,T}( \hat{V}_{t}^{X})-u}{h}, \frac{G_{n,T}( \hat{V}_{t}^{Y})-w}{h} \right) \mathrm{d}w \mathrm{d}t, \\
\widehat{ \partial_{v}C}(u,v) &= \frac{1}{hT} \int_{0}^{T} \int_{0}^{u}K \left( \frac{\widehat{F}_{n,T}( \hat{V}_{t}^{X})-w}{h}, \frac{G_{n,T}( \hat{V}_{t}^{Y})-v}{h} \right) \mathrm{d}w \mathrm{d}t.
\end{aligned}
\end{equation}
Here, $K$ is defined on a compact subset of $\mathbb{R}^{2}$, Lipschitz continuous, and such that
\begin{equation}
\int_{ \mathbb{R}^{2}}K(x,y)(x^{2}+y^{2}) \mathrm{d}x \mathrm{d}y \leq C \qquad \text{and} \qquad \int_{ \mathbb{R}^{2}}K(x,y) \mathrm{d}x \mathrm{d}y = 1.
\end{equation}
This is a common regularity condition for kernel density estimation. Note that $K$ is not required to be symmetric.

To show consistency of $\widehat{ \partial_{u}C}(u,v)$ and $\widehat{ \partial_{v}C}(u,v)$, we take a slight detour by forming a preliminary nonparametric estimator of the copula density, i.e. $c(u,v) = \partial_{u} \partial_{v} C(u,v)$,
\begin{equation} \label{equation:copula-density-kernel}
\hat{c}_{n,T}(u,v) = \frac{1}{Th^{2}} \int_{0}^{T} K \left( \frac{\widehat{F}_{n,T}( \hat{V}_{t}^{X})-u}{h}, \frac{\widehat{G}_{n,T}( \hat{V}_{t}^{Y})-v}{h} \right) \mathrm{d}t.
\end{equation}
The following condition is sufficient to get consistency of $\hat{c}_{n,T}(u,v)$.

\begin{asu} \label{asu6}
$H_{T}(x,y)$ is three times continuously differentiable on $\mathbb{R}_{+}^{2}$.
\end{asu}
In other words, we implicitly assume that the empirical copula density $c_{T}(u,v)$ exists and is itself continuously differentiable on $[0,1]^{2}$.

\begin{corollary} \label{cor2}
We suppose that Assumptions \ref{assumption:jump-activity} - \ref{asu6} hold. As $\Delta_{n} \rightarrow 0$, $T \rightarrow \infty$, $h \rightarrow 0$, such that
\begin{equation}
\quad \frac{d_{n}}{h^{3}} \vee \frac{1}{ \sqrt{T}h^{3}} \rightarrow 0,
\end{equation}
where $d_{n}$ is defined in \eqref{equation:weird-rate-condition}, for any $(u, v) \in [0,1]^{2}$ it holds that
\begin{eqnarray} \label{Sigma_{n}T}
\hat{c}_{n,T}(u, v) \overset{ \mathbb{P}}{ \longrightarrow} c(u,v).
\end{eqnarray}
\end{corollary}
The partial derivatives of $C(u,v)$ can be written as $\partial_{u}C(u,v) = \int_{0}^{v} c(u,w) \mathrm{d}w$ and $\partial_{v} C(u,v) = \int_{0}^{u} c(w,v) \mathrm{d}w$ by Assumption \ref{asu6}. Hence, $\widehat{ \partial_{u} C}(u,v)$ and $\widehat{ \partial_{v} C}(u,v)$ are Riemann approximations of these integrals based on the kernel density estimator $\hat{c}_{n,T}(u,v)$. Therefore, the consistency of the partial derivatives in \eqref{estparder} follows directly, namely
\begin{equation} \label{equation:partial_C}
\widehat{ \partial_{u} C}(u,v) \overset{ \mathbb{P}}{ \longrightarrow} \partial_{u}C(u,v) \quad \text{and} \quad \widehat{ \partial_{v}C}(u,v) \overset{ \mathbb{P}}{ \longrightarrow} \partial_{v} C(u,v).
\end{equation}
In view of \eqref{equation:avarC} and \eqref{equation:partial_C}, we define the following plug-in estimator for the covariance functional of the realized copula of volatility, $\mathrm{avar}_{ \mathcal{C}}(u,v,u',v')$:
\begin{equation} \label{escov_C}
\widehat{ \mathrm{avar}}_{ \mathcal{C}}(u,v,u',v') = \hat{ \mathbf{g}}(u,v)^{ \top} \hat{ \Gamma}(u,v,u',v') \hat{ \mathbf{g}}(u',v'),
\end{equation}
where
\begin{equation}
\begin{aligned}
\hat{ \mathbf{g}}(u,v) = \left[1, -\widehat{ \partial_{u}C}(u,v), -\widehat{ \partial_{v}C}(u,v) \right]^{ \top}, \\[0.25cm]
\hat{\Gamma}(u,v,u',v') = \text{cov} \big( \hat{ \mathbf{Z}}(u,v), \hat{ \mathbf{Z}}(u',v') \big),
\end{aligned}
\end{equation}
and
\begin{equation}
\hat{ \mathbf{Z}}(u,v) = \left[ \mathcal{G} \big(Q_{n,T}^{X}(u),Q_{n,T}^{Y}(v) \big), \mathcal{G} \big(Q_{n,T}^{X}(u), \infty \big), \mathcal{G} \big( \infty, Q_{n,T}^{Y}(v) \big) \right]^{ \top}
\end{equation}
is a three-dimensional, mean zero Gaussian vector whose covariance is estimated via the long-run covariance functional in \eqref{escov_G}. Then, by dominance of convergence in probability,
\begin{equation}
\frac{ \sqrt{T} \left( \widehat{C}_{n,T}(u,v) - C(u,v) \right)}{ \sqrt{ \widehat{ \mathrm{avar}}_{ \mathcal{C}}(u,v,u,v)}} \overset{d}{ \longrightarrow} N(0,1).
\end{equation}

\subsection{Uniform confidence bands}

The pointwise inference derived above ensures asymptotically correct coverage and testing size for each fixed value of $(u,v)$, or $(x,y)$, but it conceals that the functional CLT established in Theorem \ref{thm2} offers stronger control of the measurement errors. In this subsection, we exploit this observation to construct uniform confidence bands for the volatility copula.

\begin{corollary} \label{corband}
Under the assumptions of Theorem \ref{thm2}:
\begin{equation} \label{equation:sup}
\sup_{(u,v) \in[0,1]^{2}} \sqrt{T} \left( \widehat{C}_{n,T}(u,v) - C(u,v) \right) \overset{d}{ \longrightarrow }\sup_{(u,v) \in[0,1]^{2}}{\mathcal{C}(u,v)},
\end{equation}
where $\mathcal{C}$ is the Gaussian process defined in Theorem \ref{thm2}.
\end{corollary}
Corollary \ref{corband} employs the continuous mapping theorem for functional convergence in law, since the supremum functional is continuous in the uniform metric. However, the limiting distribution of the supremum point is non-standard, so we propose a simulation procedure to construct valid confidence bands.

We define the $(1- \alpha)$-quantile of the right-hand side in \eqref{equation:sup} as follows:
\begin{equation} \label{equation:quantile}
q_{1- \alpha}^{CB} = \inf \left\{c>0: \mathbb{P} \left( \sup_{(u,v) \in[0,1]^{2}}| \mathcal{C}(u,v)| \leq c \right) \geq 1- \alpha \right\}.
\end{equation}
Thus, it follows that as $T \rightarrow \infty$,
\begin{equation}
\mathbb{P} \left( \sup_{(u,v) \in[0,1]^{2}} \bigl| \widehat{C}_{n,T}(u,v) - C(u,v) \bigr| \leq \frac{q_{1-\alpha}^{CB}}{ \sqrt{T}} \right) \rightarrow 1- \alpha,
\end{equation}
so an asymptotic $(1- \alpha)$ uniform confidence band for $C(u,v)$ is given by
\begin{equation}
\left\{ \left[ \widehat{C}_{n,T}(u,v) \pm \frac{q_{1-\alpha}^{CB}}{ \sqrt{T}} \right], \quad \mbox{for all }
(u,v) \in [0,1]^{2} \right\}.
\end{equation}
To compute the constant $q_{1- \alpha}^{CB}$, we follow a Monte Carlo approach, which builds on a consistent estimator of the covariance function of $\mathcal{C}(u,v)$. To implement it, we discretize $[0,1]^{2}$ into an $m \times m$ grid $\left \{(u_{i}, v_{j}) \right\}_{i,j=1}^{m}$ and form the corresponding $m^{2} \times m^{2}$ covariance matrix, $\widehat{ \mathrm{avar}}_{ \mathcal{C}}$. Based on the Cholesky decomposition of $\widehat{ \mathrm{avar}}_{ \mathcal{C}} = LL^{ \top}$, we repeat the following procedure for $b = 1, \ldots, B$:
\begin{enumerate}
\item Draw $z^{(b)} \overset{\text{i.i.d.}}{ \sim} N(0, I_{m^{2}})$.
\item Set $\mathcal{Z}^{(b)} = Lz^{(b)}$, such that $\mathcal{Z}^{(b)} \sim N(0, \widehat{ \mathrm{avar}}_{ \mathcal{C}})$.
\item Compute $S_{(b)} = \max_{1 \leq k \leq m^{2}} \bigl| \mathcal{Z}^{(b)}_{k} \bigr|$.
\item Retrieve the $(1- \alpha)$-quantile of $\left \{S_{(1)}, \ldots, S_{(B)} \right\}$:
\begin{equation}
\widehat{q}_{1- \alpha}^{CB} = S_{(\lceil(1- \alpha)B \rceil)}.
\end{equation}
\item Construct the feasible uniform confidence band:
\begin{equation}
\left\{ \left[  \widehat{C}_{n,T}(u,v) \pm \frac{ \widehat{q}_{1- \alpha}^{CB}}{ \sqrt{T}} \right], \quad \mbox{for all }(u,v) \in [0,1]^{2} \right\}.
\end{equation}
\end{enumerate}
The validity of this resampling procedure is formally established by the next proposition, which is a functional version of Slutsky's theorem.

\begin{pro} \label{confidenceband}
We assume the conditions of Theorem \ref{thm2} and Corollary \ref{cor1} - \ref{cor2} hold, such that
\begin{equation}
\sqrt{T} \left( \widehat{C}_{n,T}(u,v) - C(u,v) \right) \Rightarrow \mathcal{C},
\end{equation}
where $\mathcal{C}$ is a centered Gaussian process with continuous covariance function $\mathrm{avar}_{ \mathcal{C}}(u,v,u,v)$ bounded away from zero, i.e. $\inf_{(u,v) \in [0,1]^{2}} \mathrm{avar}_{ \mathcal{C}}(u,v,u,v) > 0$. Also, let $\widehat{ \mathrm{avar}}_{ \mathcal{C}}$ be the plug-in estimator defined in \eqref{escov_C}, which consistently estimates the covariance function of $\mathcal{C}$ at the grid points. Conditionally on $\widehat{ \mathrm{avar}}_{ \mathcal{C}}$, let $( \widehat{ \mathcal{C}}_{n,T})$ be a centered Gaussian processes with covariance matrix $\widehat{ \mathrm{avar}}_{ \mathcal{C}}$. Write
\begin{equation}
\widehat{q}_{1- \alpha}^{CB} = \inf \left\{c>0: \mathbb{P} \left( \sup_{i,j} \big| \widehat{ \mathcal{C}}_{n,T}(u_{i},v_{j}) \big| \leq c \mid \widehat{ \mathrm{avar}}_{ \mathcal{C}} \right) \geq 1-\alpha \right\},
\end{equation}
Then, as the grid is made arbitrarily fine, i.e. $m \rightarrow \infty$, with $\Delta_{n} \rightarrow 0$ and $T \rightarrow \infty$, for any constant $c \geq 0$, we deduce that
\begin{equation}
\mathbb{P} \left( \sup_{(u, v) \in [0,1]^{2}} \big| \widehat{ \mathcal{C}}_{n,T}(u, v) \big| \leq c \mid \widehat{ \mathrm{avar}}_{ \mathcal{C}} \right) \overset{ \mathbb{P}}{ \longrightarrow}
\mathbb{P} \left( \sup_{(u,v) \in [0,1]^{2}} \left| \mathcal{C}(u,v) \right| \leq c \right),
\end{equation}
\end{pro}
It follows from the proposition that $\widehat{q}_{1-\alpha}^{CB} \overset{ \mathbb{P}}{ \longrightarrow} q_{1-\alpha}^{CB}$. Therefore, the uniform confidence band constructed from $\widehat{q}_{1- \alpha}^{CB}$ has the correct asymptotic coverage:
\begin{equation}
\mathbb{P} \left( \sup_{(u,v) \in[0,1]^{2}} \bigl| \widehat{C}_{n,T}(u,v) - C(u,v) \bigr| \leq \frac{ \widehat{q}_{1- \alpha}^{CB}}{ \sqrt{T}} \right) \rightarrow 1- \alpha.
\end{equation}

\subsection{Goodness-of-fit testing}

The previous analysis can also be extended to construct a goodness-of-fit test of the volatility copula function, which is another prominent application of our functional convergence results. In particular, we introduce a null hypothesis and an alternative hypothesis:
\begin{equation} \label{equation:hypothesis}
\mathcal{H}_{0}: C = C_{0} \quad \text{and} \quad \mathcal{H}_{1}: C \neq C_{0}.
\end{equation}
Here, equality is almost everywhere on $[0,1]^{2}$.

To construct a testing procedure, we let
\begin{equation}
G_{n,T}(u,v) \equiv \sqrt{T} \left( \widehat{C}_{n,T}(u,v) - C_{0}(u,v) \right).
\end{equation}
We now arrive at the following statement, which again follows from the continuous mapping theorem for functional convergence applied to Theorem \ref{thm2}.
\begin{corollary} \label{cortest}
Under the assumptions of Theorem \ref{thm2} and conditional on $\mathcal{H}_{0}$:
\begin{equation} \label{eq:limitsup}
\lVert G_{n,T} \rVert_{L_{2}}^{2} \overset{d}{ \longrightarrow} \lVert \mathcal{C} \rVert_{L_2}^{2}.
\end{equation}
\end{corollary}

The limiting distribution in Corollary \ref{cortest} (a sort of weighted infinite mixture of $\chi^{2}$ random variables) is again not standard, so we employ a Monte Carlo procedure, closely related to the one outlined above, to form a test of $\mathcal{H}_{0}$ versus $\mathcal{H}_{1}$. In particular, we let $q_{1- \alpha}^{GF}$ be a consistent estimator of the $(1- \alpha)$-quantile of the distribution of $\lVert \mathcal{C} \rVert_{L_{2}}^{2}$, and we define a rejection region:
\begin{equation}
\mathcal{R} = \left\{ \lVert G_{n,T} \rVert_{L_{2}}^{2} \geq q_{1- \alpha}^{GF} \right\}.
\end{equation}
Then, in view of Corollary \ref{cortest},
\begin{equation}
\mathbb{P}( \mathcal{R} \mid \mathcal{H}_0) \rightarrow \alpha \quad \text{and} \quad \mathbb{P}( \mathcal{R} \mid \mathcal{H}_{1}) \rightarrow 1.
\end{equation}
To obtain $q_{1-\alpha}^{GF}$, we approximate the distribution of the random variable $\lVert \mathcal{C} \rVert_{L_{2}}^{2}$ by a weighted Riemann sum on a, possibly unequally spaced, grid $\{(u_{i},v_{j}) \}_{i,j=1}^{m} \subset [0,1]^{2}$, i.e.
$\sum_{1 \leq i,j \leq m} w_{ij} \mathcal{C}(u_{i},v_{j})^{2}$, where $w_{ij}$ are weights induced by the grid and $\sum_{i,j} w_{ij} = 1$. Note that $\{ \mathcal{C}(u_{i}, v_{j}) \}_{i,j=1}^{m}$ is a sequence of $m^{2}$ normal random variates with explicit covariance matrix given by \eqref{equation:Gamma} following Theorem \ref{thm2}, which we denote by $\Gamma$. Therefore,
\begin{equation}\label{equation:chi-square}
\sum_{i,j} w_{ij} \mathcal{C}(u_{i}, v_{j})^{2} \overset{d}{=} \sum_{k=1}^{m^{2}} \pi_{k} \chi_{k}^{2},
\end{equation}
where $\chi_{k}^{2}$ are i.i.d. $\chi^{2}(1)$ random variables and $\pi_{k}$ are the eigenvalues of $W^{1/2} \Gamma W^{1/2}$ with $W = \mathrm{diag} \big( \{w_{ij} \}_{i,j=1}^{m} \big)$, for $k = 1, \dots, m^{2}$.

We draw repeated observations of the right-hand side of \eqref{equation:chi-square} based on the eigenvalues from the estimator $W^{1/2} \widehat{ \Gamma} W^{1/2}$ to simulate the distribution of $\lVert G_{n,T} \rVert_{L_{2}}^{2}$ under $\mathcal{H}_{0}$, from which an appropriate quantile can be extracted.\footnote{Since $\widehat{ \Gamma}$ is an estimator of a covariance matrix, it can possess negative eigenvalues. We follow \citet{andersen-su-todorov-zhang:24a} and retain only those terms in the sum associated with positive eigenvalues.}

To estimate the covariance matrix $\Gamma$, we rearrange the pairs $\{(u_{i}, v_{j})\}_{i,j=1}^{m}$ as $\{(u_{k}', v_{k}')\}_{k=1, \dots, m^{2}}$ and compute $\widehat{ \Gamma}_{k \ell} =  \hat{g}_{k}^{ \top} \hat{A}_{k \ell} \hat{g}_{ \ell}$, where
\begin{equation}
\begin{aligned}
\hat{g}_{k} &= \left[1, - \widehat{ \partial_{u}C}(u_{k}, v_{k}), - \widehat{ \partial_{v}C}(u_{k},v_{k}) \right]^{ \top}, \\
\hat{A}_{kl} &= \text{cov} \big( \hat{ \mathbf{Z}}(u_{k},v_{k}), \hat{ \mathbf{Z}}(u_{l},v_{l}) \big),
\end{aligned}
\end{equation}
and
\begin{equation}
\hat{ \mathbf{Z}}_{k} = \left[ \mathcal{G} \big(Q_{n,T}^{X}(u_{k}),Q_{n,T}^{Y}(v_{k}) \big), \mathcal{G} \big(Q_{n,T}^{X}(u_{k}), \infty \big), \mathcal{G} \big( \infty, Q_{n,T}^{Y}(v_{k}) \big) \right]^{ \top}.
\end{equation}
It is also possible to test a composite null hypothesis:
\begin{eqnarray*}
\mathcal{H}_0: C \in \mathcal{S} \qquad \text{versus} \qquad \mathcal{H}_{a}: C \notin \mathcal{S},
\end{eqnarray*}
where $\mathcal{S} = \left\{C_{ \theta} : \theta \in \Theta \subseteq \mathbb{R}^{p} \right\}$ is a statistical model for the copula function, and $\theta$ is a $p$-dimensional parameter vector.

In practice, the true value of the copula parameter, $\theta_{0}$, is unknown. A complication occurs if the parameter is estimated from the full sample, in which case the estimator, $\hat{ \theta}$, is typically $\sqrt{T}$-consistent matching the convergence rate of $\widehat{C}_{n,T}(u,v)$. Then, $G_{n,T}(u,v) = \sqrt{T}( \widehat{C}_{n,T}(u,v) - C_{ \hat{ \theta}}(u,v)) = \sqrt{T}(\widehat{C}_{n,T}(u,v) - C_{ \theta_{0}}(u,v)) - \sqrt{T}( \hat{ \theta} - \theta_{0})^{ \top} \partial_{ \theta} C_{ \theta_{0}}(u,v) + o_{p}(1)$, where $\partial_{ \theta} C_{ \theta_{0}}(u,v) =  \partial_{ \theta} C_{ \theta}|_{ \theta = \theta_{0}}$, so $G_{n,T}$ converges in law to a Gaussian process \emph{plus} a bias term that is a linear functional of $\sqrt{T}( \hat{ \theta} - \theta_{0}) \overset{d}{ \longrightarrow} \mathcal{W} \sim N(0, \Sigma_{ \theta})$. Hence, the $L^{2}$-statistic $\lVert G_{n,T} \rVert_{L_{2}}^{2} \overset{d}{ \longrightarrow} \lVert \mathcal{C} - \langle \partial_{ \theta} C_{ \theta_{0}}, \mathcal{W} \rangle \rVert_{L_2}^{2}$ is a quadratic form of a Gaussian process, which depends on unknown quantities. In other words, the asymptotic distribution is non-pivotal, as it depends on the parameter vector $\theta_{0}$, the gradient $\partial_{ \theta} C_{ \theta_{0}}(u,v)$ and the law of $(\mathcal{W}, \mathcal{C})$. However, we can easily resample the empirical process and the plug-in estimator of the covariance function using a parametric or multiplier bootstrap under $C_{\hat{ \theta}}$ to approximate the distribution of $\lVert G_{n,T} \rVert_{L_{2}}^{2}$ and construct critical values without analytically having to remove the dependence on $\theta_{0}$.

On the other hand, to avoid the effect of $\hat{\theta}$, we can estimate the copula parameter such that $\hat{\theta} - \theta_{0}$ achieves a convergence rate faster than $\widehat{C}_{n,T}(u,v) - C_{ \theta_{0}}(u,v)$.\footnote{This can be accomplished, for example, by reserving data in a shorter time interval, say $[0, T^{\eta}]$ with $0 < \eta < 1$, to construct $\widehat{C}_{n,T}$, but still employing data over the whole time interval $[0,T]$ to get $\hat{\theta}$.} In this case, the test statistic is asymptotically pivotal and we can replace $C_{0}$ by $C_{\hat{\theta}}$ in the previous testing procedure with a single null hypothesis and attain a valid inference that does not involve nuisance parameters and is shape-only. Hence, it suffices to prepare a single table of critical values.

\section{Simulation analysis} \label{section:simulation}

In this section, we conduct a Monte Carlo study to explore the finite sample properties of our realized estimator of the empirical copula of stochastic volatility. We restrict attention to the bivariate setting and assume that a pair of asset log-price processes follow the motion
\begin{equation} \label{equation:Xsim}
\mathrm{d}X_{i,t} = \mu_{i} \mathrm{d}t + \sqrt{V_{i,t}} \mathrm{d}W_{i,t},
\end{equation}
where
\begin{equation} \label{equation:Vsim}
\mathrm{d} \log V_{i,t} = \kappa \left( \eta_{i} - \log V_{i,t} \right) \mathrm{d}t + \beta_{i} \mathrm{d}B_{i,t},
\end{equation}
for $i = 1$ and $2$.

In the above, $W$ and $B$ are Brownian motions with leverage correlation $\text{corr}[\mathrm{d}W_{i,t} \mathrm{d}B_{i,t}] = \rho = -\sqrt{0.5}$. We assume a common speed of mean reversion $\kappa = 10$. The other parameters are $\mu_{1} = 0.10$, $\eta_{1} = \log(0.05)$, $\beta_{1} = 3$, $\mu_{2} = 0.05$, $\eta_{2} = \log(0.01)$, and $\beta_{2} = 3$. This design is intended to represent a risky asset with an expected return and volatility of 10\% and 20\% and a near risk-free asset with an expected return and volatility of 5\% and 10\%. It emulates our empirical work in Section \ref{section:empirical}, where we look at high-frequency data from stock index and medium-duration treasury bond futures contracts.

A routine calculation for the stationary Gaussian Ornstein-Uhlenbeck process shows that the unconditional distribution of $\log V_{i,t}$ has the form
\begin{equation}
\log V_{i,t} \sim N \left( \eta_{i}, \frac{ \beta_{i}^{2}}{2 \kappa} \right).
\end{equation}
We denote the distribution function as $F_{ \log V_{i}}(v)$, for $- \infty < v < \infty$.

The univariate log-variance processes are tied together with a Gumbel copula:
\begin{equation} \label{equation:Gumbel}
C(u,v) = \exp \left(- \left[(- \log u)^{ \theta} + (- \log v)^{ \theta} \right]^{1 / \theta} \right), \quad \text{for } 0 \leq u,v \leq 1,
\end{equation}
which is parameterized by $\theta \in [1, \infty)$. A larger value indicates a stronger relationship. The independence copula---$C(u,v) = uv$---corresponds to $\theta = 1$, while the comonotonicity copula---$C(u,v) = \min(u,v)$---induced by the upper Fr\'{e}chet-Hoeffding bound is the limit as $\theta \rightarrow \infty$. In our empirical application, we explore the properties of this particular copula, because our goodness-of-fit test suggests that it provides an excellent description of the dependence between the log-variance of the equity and treasury bond market. In line with our maximum likelihood estimate from that section, we set $\theta = 2$.

To construct a realization of this process, we begin by drawing from the Gumbel copula to generate a pair of dependent random variables $(U_{1,0}, U_{2,0})$, where $U_{i,0} \sim U(0,1)$. These serve to initialize the log-variance processes via $\log V_{i,0} = F_{\log V_i}^{-1}(U_{i,0})$. We then apply an Euler approximation with constant step size $\Delta$ to iteratively generate the sample path of \eqref{equation:Xsim}--\eqref{equation:Vsim} forward in time, using the Metropolis--Hastings algorithm to ensure that each proposal update is consistent with the target distribution.\footnote{The symbol $\Delta$ does double-duty in this context. In isolation it is the time increment $\Delta = t - s$, but when it is applied to a stochastic process $Y$, it represents the difference operator $\Delta Y_{i,t} = Y_{i,t} - Y_{i,s}$.}

In the ``continuous-time'' version of the model, we set $\Delta = 1/N$ with $N = 23{,}400$. We assume data is captured over a time interval $[0,T]$ with $T = 50, 100, \dots, 2000$, but the log-price process is only observed at a much coarser equidistant grid $t_{i} = i/n$ with $n = 39, 78, 390$, for $i = 0, 1, \dots, nT$. This design is intended to imitate a security that is traded in a 6.5 hour window each day and whose price is recorded at the ten-, five-, and one-minute frequency over at most an eight-year period. This resembles the dataset analyzed in our empirical application. To estimate the latent spot volatility, associated with each $n$ we employ a bandwidth of $h_{n} = 36, 48, 120$. With the above interpretation, this amounts to a six-, four-, and two-hour window. Hence, as $n$ increases we collect a larger amount of high-frequency data while reducing the time span of the estimation window, as stipulated by the rate conditions in the asymptotic theory. These choices follow \citet{li-todorov-tauchen:13a} and \citet{christensen-thyrsgaard-veliyev:19a}.\footnote{In unreported results that are available at request, we conducted a more extensive set of numerical experiments with other bandwidth choices. This showed that the average squared estimation error of the spot volatility estimator was basically flat over an extended region, so long as $h_{n}$ was not too small or too large.}

A sample path of each log-variance process, together with the log-realized variance, is illustrated in Panel A of Figure \ref{figure:illustration}. In Panel B, we report a scatter of 500 randomly selected pairs of the foregoing series, after transforming them through their empirical distribution functions. We contrast this with selected contours of the probability density function of the Gumbel copula. The latter are increasing as we move toward the 45$^{\circ}$-line and eventually disconnect at the lower-left and upper-right corners. The main takeaway is that while the log-variance observations are dispersed as expected, the estimation error in the log-realized variance shifts several observations into regions where the Gumbel copula has almost no concentration of mass.

\begin{figure}[t!]
\begin{center}
\caption{Illustration of bivariate SV model with Gumbel copula.}
\label{figure:illustration}
\begin{tabular}{cc}
\small{Panel A: Log-variance.} & \small{Panel B: Copula.} \\
\includegraphics[height=8cm,width=0.47\textwidth]{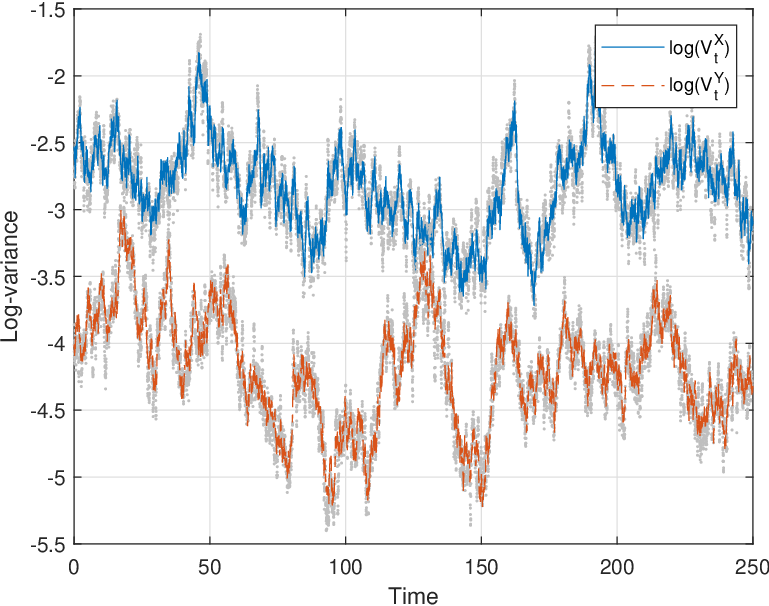} &
\includegraphics[height=8cm,width=0.47\textwidth]{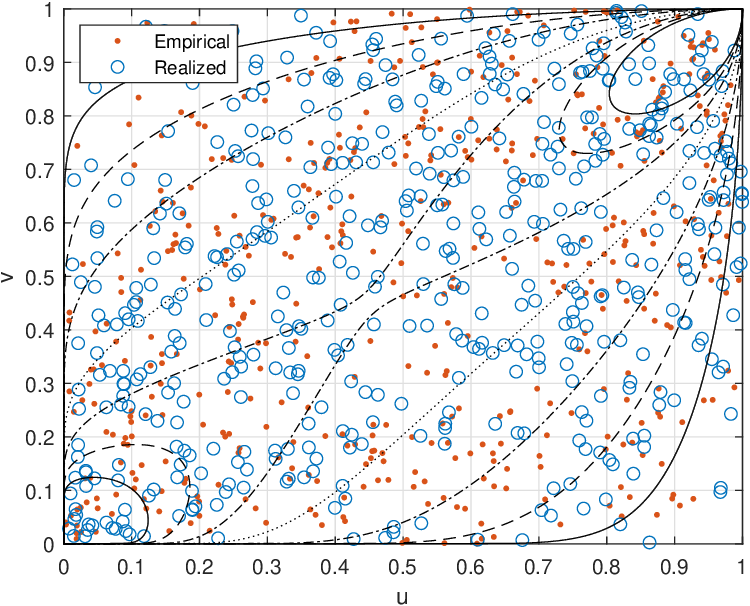}
\end{tabular}
\begin{footnotesize}
\parbox{\textwidth}{\emph{Note.} In Panel A, we show a realization of the log-variance processes of $X$ and $Y$ together with the associated log-realized variance series. The latter appear as a grey-shaded color in the background. In Panel B, we construct a scatter of 500 randomly selected pairs of the foregoing series, after transforming them through their empirical distribution function. We contrast this with selected contours of the probability density function of the Gumbel copula, which are increasing as we move toward the 45$^{\circ}$-line.}
\end{footnotesize}
\end{center}
\end{figure}

The outcome of the full-blown simulation analysis is summarized in Figure \ref{figure:copula}.  We perform $M = 1{,}000$ repetitions in total. In Panel A, we construct a decile plot of the bivariate distribution function of the Gumbel copula, which we contrast against the empirical copula of volatility---based on $T = 2{,}000$---along with the estimated one recovered by the realized version. We observe that even with a conservative amount of high-frequency data, the empirical copula is close to its stationary counterpart with only some minor deviations. In contrast, the contours of the realized version are generally farther away and located to the northeast of the empirical copula. The explanation is that the sampling variation in the realized variances render them less correlated and shifts mass in the direction of the independence copula.\footnote{The level curves of the independence copula are given by $u = c/v$ for $c \in (0,1)$.} As $n$ increases, however, the realized copula of volatility steadily trends toward the empirical counterpart, which is a validation of the asymptotic theory developed in Section \ref{section:theory}.

\begin{figure}[t!]
\begin{center}
\caption{Simulation result.}
\label{figure:copula}
\begin{tabular}{cc}
\small{Panel A: Deciles of copula.} & \small{Panel B: RMSE.} \\
\includegraphics[height=8cm,width=0.47\textwidth]{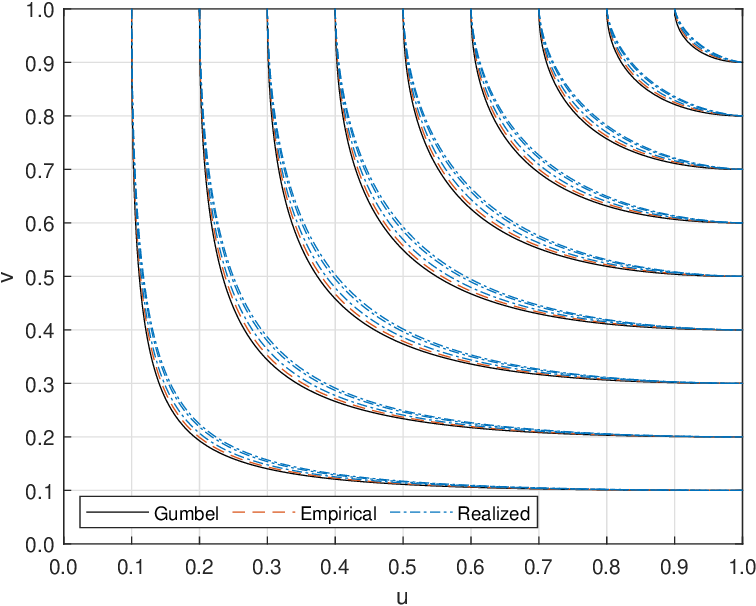} &
\includegraphics[height=8cm,width=0.47\textwidth]{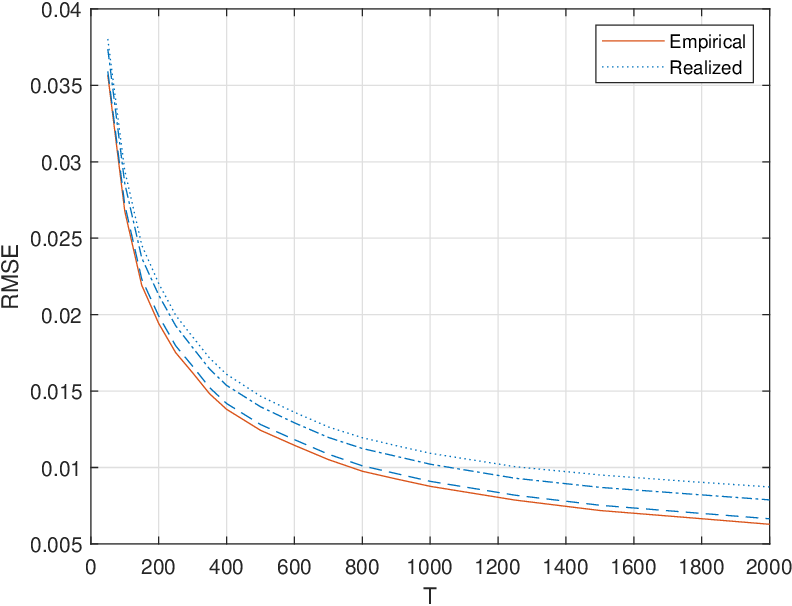}
\end{tabular}
\begin{footnotesize}
\parbox{\textwidth}{\emph{Note.} In Panel A, we construct a decile plot of the distribution function of the Gumbel copula, which we compare to the deciles of the empirical copula of volatility (based on $T = 2{,}000$) and the realized version (with $n = 30$ to $n = 390$ as we approach the empirical copula). In Panel B, we plot the RMSE of the empirical and realized copula of volatility in terms of their estimation accuracy for the limiting Gumbel copula.}
\end{footnotesize}
\end{center}
\end{figure}

In Panel B of Figure \ref{figure:copula}, we report the root mean squared error (RMSE):
\begin{equation}
\mathrm{RMSE}(x) = \frac{1}{M} \sum_{m=1}^{M} \sqrt{ \int_{0}^{1} \int_{0}^{1} [ C_{m}^{x}(u,v) - C(u,v) ]^{2} \mathrm{d}u \mathrm{d}v},
\end{equation}
where $C(u,v)$ is the Gumbel copula given by the expression in \eqref{equation:Gumbel} and $C_{m}^{x}(u,v)$ represents either the empirical ($x = \mathrm{Empirical}$) or realized ($x = \mathrm{Realized}$) estimate during the $m$th simulation. The figure reveals that, as expected, the RMSE of the empirical copula of volatility approaches zero, as $T \rightarrow \infty$, which is consistent with our theoretical exploration. It takes about a quadrupling of the long-span dimension to reduce the RMSE in half, which reflects its convergence rate. Turning to the realized copula of volatility, we see that the RMSE is generally higher, and while it too drops as $T$ increases, it tends to flatten out strictly above zero. This is because the estimation error does not vanish with a fixed $n$. Boosting $n$, however, lessens the discretization error, and by the time that $n = 390$ the realized version is within proximity of the empirical copula (its target in the in-fill limit as $n \rightarrow \infty$).

To gauge the functional CLT for the realized copula of volatility as an approximation to its finite sample behavior, we next explore the distribution of the asymptotic pivot:
\begin{equation}
Z_{ \mathcal{C}} = \frac{ \sqrt{T} \big( \widehat{C}_{n,T}(u,v) - C(u,v) \big)}{ \sqrt{ \mathrm{avar}_{ \mathcal{C}}(u,v,u,v)}} \overset{d}{ \longrightarrow} N(0,1),
\end{equation}
for $u = v = 0.10, 0.25, 0.50, 0.75, 0.90$.

\begin{figure}[t!]
\begin{center}
\label{figure:inference_of_realized_copula}
\caption{Finite sample distribution of studentized realized copula of volatility.}
\begin{tabular}{cc}
\text{Panel A: Kernel density estimate.} & \text{Panel B: Q--Q plot.} \\
\includegraphics[height=8cm,width=0.47\textwidth]{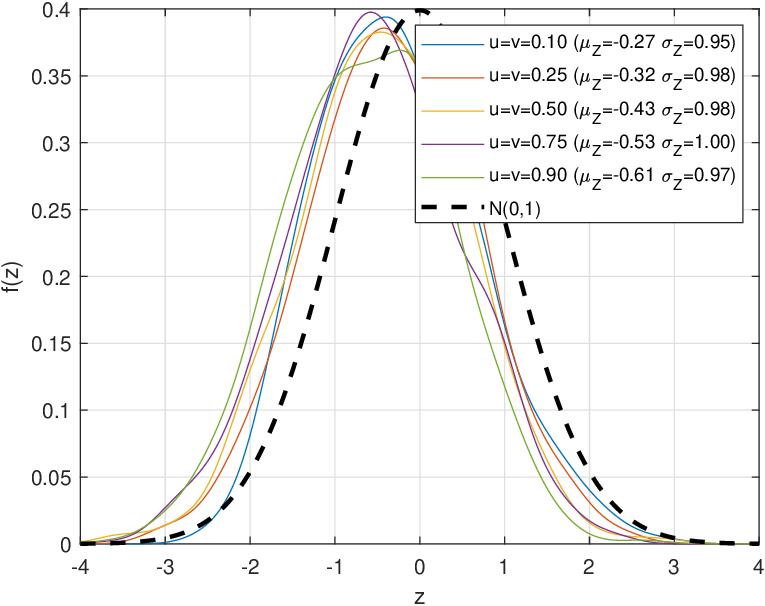} &
\includegraphics[height=8cm,width=0.47\textwidth]{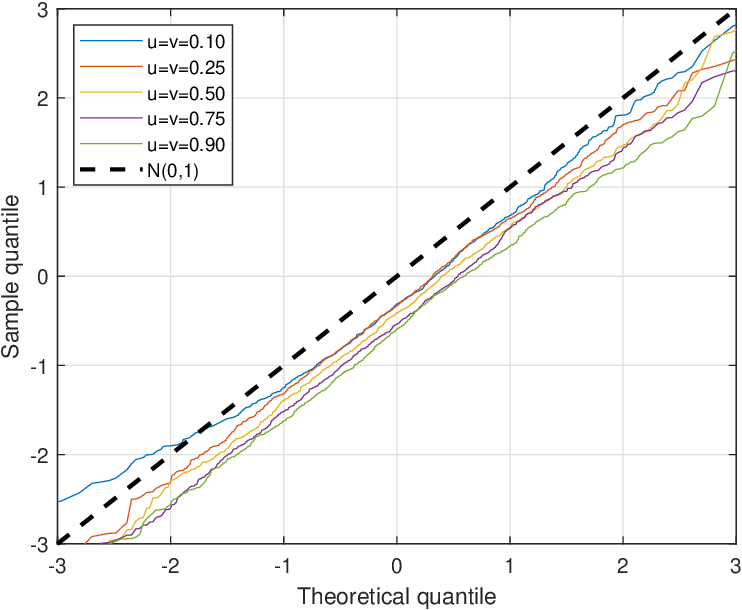}
\end{tabular}
\begin{footnotesize}
\parbox{\textwidth}{\emph{Note.} In Panel A, we report a kernel density estimate (based on a Gaussian kernel with bandwidth chosen by Silverman's rule of thumb) of the studentized statistic $Z_{ \mathcal{C}} = \sqrt{T} \left( \widehat{C}_{n,T}(u,v) - C(u,v) \right)/ \sqrt{ \mathrm{avar}_{ \mathcal{C}}(u,v,u,v)}$, for $n = 390$, $T = 2{,}000$, and $u = v = (0.10,0.25,0.50,0.75,0.90)$. A standard normal curve (dashed line) is superimposed as a visual benchmark. In each setting, we also report the sample average and standard deviation of $Z_{ \mathcal{C}}$. In Panel B, we show the corresponding normal Q--Q plot. The Monte Carlo experiment has $1{,}000$ replications.}
\end{footnotesize}
\end{center}
\end{figure}

Figure \ref{figure:inference_of_realized_copula} presents the outcome. We base the analysis on $n = 390$ and $T = 2{,}000$, which adheres to our empirical work in the next section, and the estimator of the asymptotic variance introduced in Section \ref{section:theory}. In Panel A, we show kernel density estimates, while Panel B reports the associated quantile-quantile (Q--Q) plot. We include the sample average and sample standard deviation of $Z_{ \mathcal{C}}$ for each setting. We observe that the kernel density estimates are well-aligned with the standard normal curve, and the Q--Q plots follow the $45^{o}$-line reasonably. The dispersion is almost as predicted, but the test statistic does exhibit a minuscule downward bias of about 0.2-0.5 standard deviation. Importantly, these numbers are more or less identical to those we compute for the empirical copula of volatility.\footnote{The kernel density estimates and Q--Q plots for the empirical copula of volatility are excluded from the presentation for brevity, but they can be shared at request.}

At last, we turn our attention toward the goodness-of-fit test. We calculate rejection rates of the test statistic in \eqref{eq:limitsup} from Corollary \ref{cortest} with $n = 390$ and $T = 250,500,1000,2000$. To curb the computational cost, we implement the bootstrapping algorithm for producing critical values with 5,000 draws on a $5 \times 5$ grid $(u,v) \in \{0.10, 0.25, 0.50, 0.75, 0.90 \}^{2}$, yielding 25 evaluation markers.\footnote{It is arguably preferable to use a more refined grid, or possibly even alter how the points are scattered over the domain. The intuition is that if the test statistic gets a better resolution of the realized copula of volatility, it has more scope to detect departures from the null, and this may improve rejection rates under the alternative. However, this also increases the dimensionality of the required covariance matrix, impacting memory usage and runtime speed. In line with this argument, we surmise that the power of the test reported here is conservative.}

The null hypothesis (size) is based on the Gumbel copula with $\theta = 2$, which has Kendall's tau given by $\tau = 1 - 1/ \theta$, so that our choice implies $\tau = 0.5$. The alternative hypothesis (power), is given by three distinct parametric models. The first benchmark is the independence copula, $C(u,v) = uv$. This corresponds to the absence of volatility codependency with $\tau = 0$ and represents a stark violation of $\mathcal{H}_{0}$. It provides a natural baseline for assessing the ability of the test to detect tail concentration. The second choice is the Clayton copula, $C(u,v) = \left(u^{- \alpha} + v^{- \alpha} - 1 \right)^{-1/ \alpha}$ with parameter $\alpha \geq 0$. This is another example of an Archimedean copula used to capture positive tail dependence, but in contrast to the Gumbel copula, the dependence is located in the other extreme of the distribution. The Clayton copula has $\tau = \alpha/( \alpha+2)$, and we set $\alpha = 2$---or $\tau = 0.5$---to be consistent with the Gumbel. This design intends to isolate departures from the null in the form of dependence in the opposite direction, while fixing the overall rank correlation summarized by Kendall's tau. At last, we complement the analysis with a power experiment based on an alternative member of the Gumbel class. Specifically, we examine a Gumbel copula with $\theta = 1.5$---or $\tau = 1/3$---that also has upper tail concentration, albeit slightly weaker compared to the null.

\begin{table}[ht!]
\setlength{\tabcolsep}{0.15cm}
\begin{center}
\caption{Rejection rates of the goodness-of-fit test.}
\label{table:gof}
\begin{small}
\begin{tabular}{lcccccccccccc}
\hline \hline
& \multicolumn{3}{c}{$T=250$} & \multicolumn{3}{c}{$T=500$} & \multicolumn{3}{c}{$T=1000$} & \multicolumn{3}{c}{$T=2000$} \\
\cline{2-13}
& 10\% & 5\% & 1\% & 10\% & 5\% & 1\% & 10\% & 5\% & 1\% & 10\% & 5\% & 1\% \\
\hline
\multicolumn{13}{l}{\textit{Panel A: Size}} \\
Gumbel(2) & 0.135 & 0.088 & 0.035 & 0.116 & 0.072 & 0.024 & 0.111 & 0.068 & 0.023 & 0.094 & 0.056 & 0.017 \\
\\
\multicolumn{13}{l}{\textit{Panel B: Power}} \\
Independence & 0.895 & 0.845 & 0.699 & 0.987 & 0.972 & 0.910 & 1.000 & 1.000 & 1.000 & 1.000 & 1.000 & 1.000 \\
Clayton & 0.230 & 0.167 & 0.085 & 0.330 & 0.194 & 0.087 & 0.550 & 0.312 & 0.092 & 0.936 & 0.814 & 0.371 \\
Gumbel(1.5) & 0.428 & 0.335 & 0.174 & 0.511 & 0.418 & 0.255 & 0.672 & 0.581 & 0.378 & 0.850 & 0.780 & 0.600 \\
\hline \hline
\end{tabular}
\end{small}
\smallskip
\begin{scriptsize}
\parbox{\textwidth}{\emph{Note}. We report Monte Carlo rejection rates based on 1,000 replica. The simulation design is described in the main text. Panel A with size is for the null hypothesis. Panel B reports power against the three alternatives listed in the first column. }
\end{scriptsize}
\end{center}
\end{table}

The results are reported in Table \ref{table:gof}. We inspect the nominal significance levels 10\%, 5\%, and 1\%. The size analysis is shown in Panel A. We observe that the test is to a small extent oversized for the lowest value of $T$, but it rapidly settles around the expected level. Moving to Panel B, we learn that the power versus the independence copula is almost unity across scenarios, even at the smallest level of significance. Moreover, the test has a good ability to distinguish between the Gumbel null and the Clayton alternative, especially with larger sample sizes. Lastly, while the rejection rates against the Gumbel(1.5) alternative are initially higher than those of the Clayton copula, but they increase less fast with $T$. The power function is possibly flatter in the Gumbel direction, where the functional form is identical and only the strength of the upper tail dependence is different.

Overall, our simulations suggest that the realized copula of volatility can reliably infer the latent volatility copula at a sampling frequency that is relevant for practical analysis. Moreover, the small sample properties of the standardized statistic are close to those implied by Theorem \ref{thm2}, which is assuring, because it means that confidence intervals and hypothesis tests based on the normal distribution are accurate. We also uncover that the goodness-of-fit test has appropriate size control and excellent power against many relevant alternatives. A word of caution, though. Our goodness-of-fit test merely provides evidence against the null hypothesis as a whole. Thus, a rejection of a simple null in favor of the alternative should not be interpreted as identifying the exact source of misspecification. Indeed, it can arise either because the assumed functional form of the copula is wrong, or because the family is correctly specified but the concrete parameter value is not.

\section{Empirical application} \label{section:empirical}

The forensic analysis presented here aims to determine the codependency between the volatility processes of a pair of leading financial indicators. We delve into high-frequency transaction data from futures contracts that track the aggregate U.S. equity index and treasury bond market, namely the E-mini S\&P 500 (ES) and 10-year treasury note (TY). The chosen instruments are listed on the Globex platform at the Chicago Mercantile Exchange (CME). They trade around the clock five days per week, but we restrict attention to the most active hours from 9:30am to 4:00pm EST that overlap with the NYSE trading session. The data at our disposal were purchased from Tickdata (\url{https://www.tickdata.com/}). The sample covers the period from March 18, 2010 to October 14, 2021. We remove short days associated with a reduced trading schedule, which leaves $T = 2{,}891$ days for our empirical investigation.\footnote{To construct a consecutive price series, we employ a built-in procedure in the extraction software delivered by Tickdata. It rolls the front contract over into the back contract about a week before expiration.}

\begin{table}[ht!]
\setlength{ \tabcolsep}{0.30cm}
\begin{center}
\caption{Descriptive statistics of CME futures contract data.}
\label{table:cme}
\begin{small}
\begin{tabular}{lrrrrrrrrr}
\hline \hline
& & \multicolumn{8}{c}{Log-realized variance} \\
\cline{3-10}
Ticker & $N$ & Min. & 1Q & Median & Mean & 3Q & Max. & Skewness & Kurtosis \\
\hline
ES & 338.78 & -7.72 & -5.63 & -5.00 & -4.91 & -4.29 & 1.23 & 0.59 & 3.78 \\
TY & 90.95 & -8.27 & -7.08 & -6.80 & -6.69 & -6.43 & -2.80 & 1.39 & 6.42 \\
\hline \hline
\end{tabular}
\end{small}
\smallskip
\begin{scriptsize}
\parbox{\textwidth}{\emph{Note}. Ticker is the short name of the futures contract, $N$ is the average daily number of transactions (in 1000s), whereas Min., 1Q, Median, Mean, 3Q, Max., Skewness, and Kurtosis are for the distribution of the log-realized variance over the whole sample.}
\end{scriptsize}
\end{center}
\end{table}

In Table \ref{table:cme}, we report a few descriptive statistics of the retained data. The ``$N$'' column reports the average daily number of transactions (in 1000s). These futures contracts are very liquid. However, sampling the tick-by-tick data at the maximal resolution can be detrimental for volatility estimation, because of microstructure noise, such as price discreteness and bid-ask spread \citep[e.g.,][]{hansen-lunde:06b}. To alleviate this concern, we downsample to a one-minute equidistant frequency using previous tick imputation, so there are $n = 390$ log-returns available per day to excavate the volatility process.\footnote{We follow the design of the simulation section in terms of choosing the various tuning parameters that are required to calculate the spot realized variance.}

\begin{figure}[t!]
\begin{center}
\caption{Histogram of log-realized variance.}
\label{figure:cme-histogram}
\begin{tabular}{cc}
\small{Panel A: ES.} & \small{Panel B: TY.} \\
\includegraphics[height=8cm,width=0.47\textwidth]{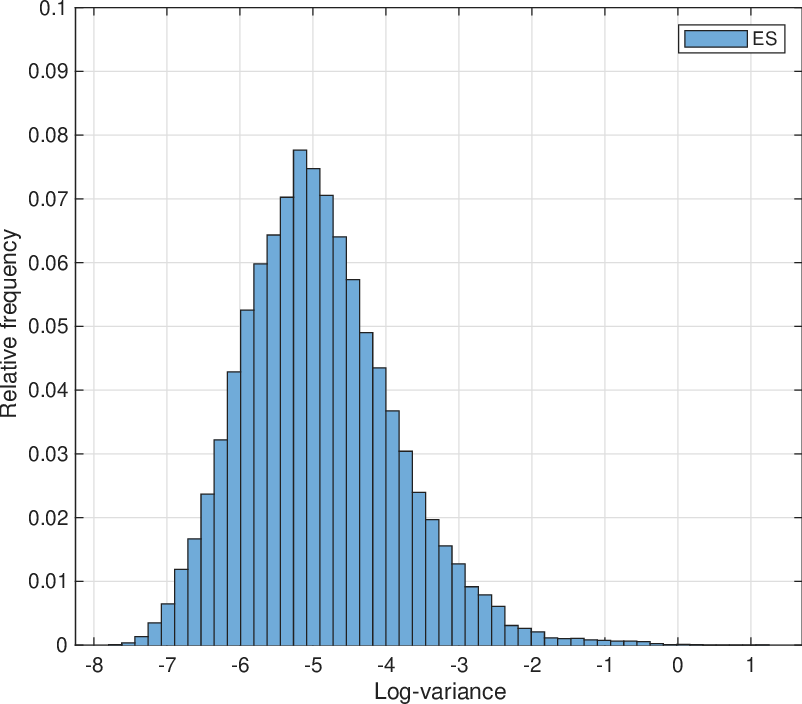} &
\includegraphics[height=8cm,width=0.47\textwidth]{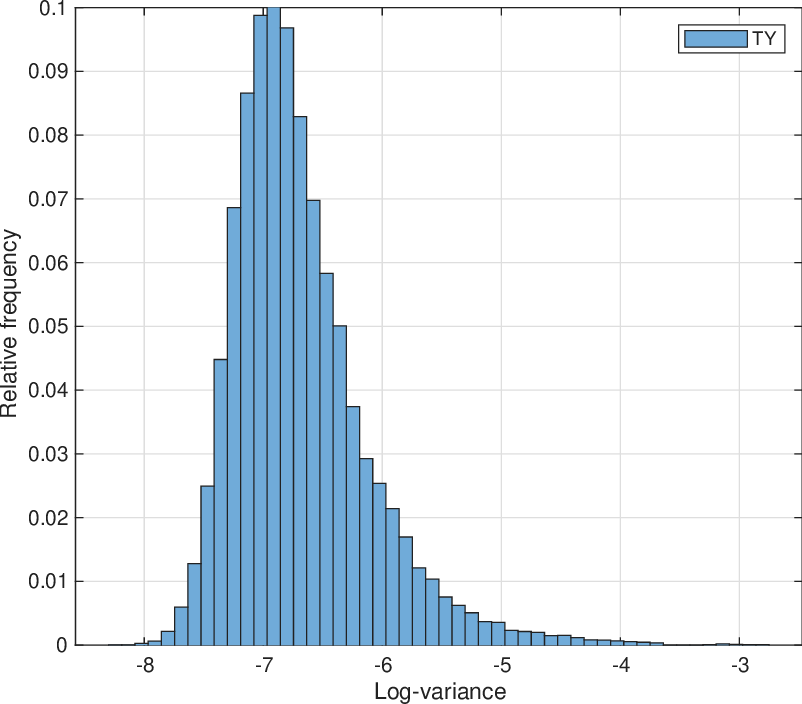}
\end{tabular}
\begin{footnotesize}
\parbox{\textwidth}{\emph{Note.} In Panel A, we show a histogram (plotted in terms of relative frequency) of the log-realized variance of ES, whereas Panel B is the associated plot for the log-realized variance of TY.}
\end{footnotesize}
\end{center}
\end{figure}

In what follows, we study the log-spot realized variance, which is less influenced by extreme outliers. We remark that by the invariance principle the copula of a random vector does not change under a strictly increasing transformation, so the realized copula of volatility is unaffected by studying the log-transform. In Figure \ref{figure:cme-histogram}, we show a relative frequency histogram for the log-realized variance of ES in Panel A, whereas Panel B is for TY. In agreement with the sample average and quantiles reported in Table \ref{table:cme}, we observe that volatility in the stock market is, on average, higher than in the bond market. Furthermore, and consistent with \citet{andersen-bollerslev-diebold-ebens:01a, andersen-bollerslev-diebold-labys:03a}, the shapes are close to the Gaussian bell curve, although both exhibit a mild positive skewness and excess kurtosis compared to the normal distribution. As evident from the right-hand side of Table \ref{table:cme}, this effect is more pronounced in TY than ES.

To provide an initial assessment on the degree of volatility codependency in these time series, Figure \ref{figure:cme-timeseries} reports the average daily log-spot realized variance for ES in Panel A and TY in Panel B. The smoothing attenuates the measurement error in the spot volatility estimator a bit to better capture any comovement. There is forceful evidence of a positive relationship, most notably during periods of elevated market distress. To support this claim, in Panel A of Figure \ref{figure:cme-variance} we draw a scatter plot of the log-realized variance of ES (on the $x$-axis) against the log-realized variance of TY (on the $y$-axis). The graph reinforces the previous impression, which is further backed by Pearson's and Kendall's sample correlation coefficients.\footnote{Kendall's tau has been transformed as $\hat{ \rho}_{ \text{K}}^{a} = \sin( \pi/2 \hat{ \rho}_{ \text{K}}^{u})$, where $a$ ($u$) denotes the adjusted (unadjusted = raw) estimator. The correction---motivated by Greiner's equality---makes it unbiased for the correlation coefficient of a random sample from a bivariate normal distribution and, hence, its scale is comparable to the Pearson's measure.}

\begin{figure}[ht!]
\begin{center}
\caption{Time series of log-realized variance.}
\label{figure:cme-timeseries}
\begin{tabular}{cc}
\small{Panel A: ES.} & \small{Panel B: TY.} \\
\includegraphics[height=8cm,width=0.47\textwidth]{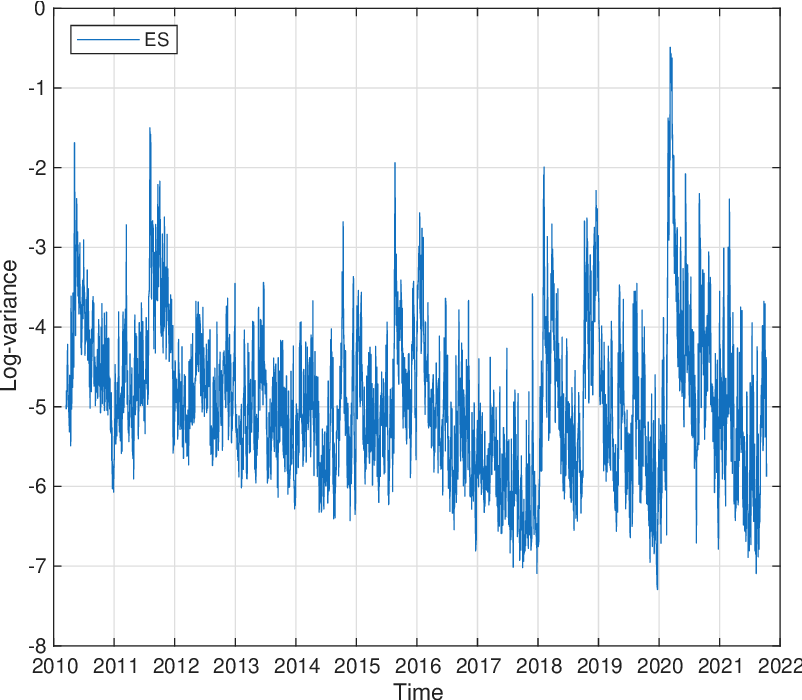} &
\includegraphics[height=8cm,width=0.47\textwidth]{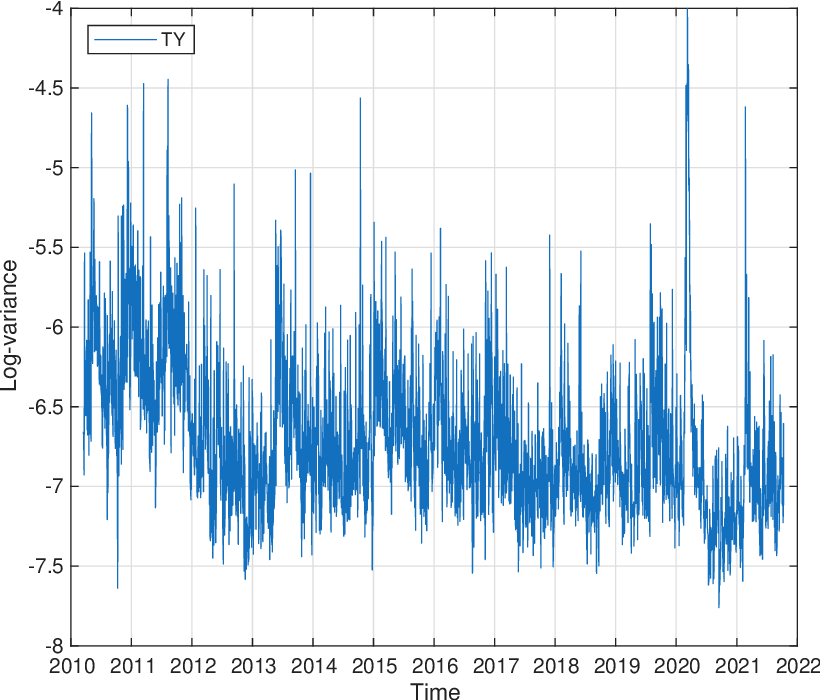}
\end{tabular}
\begin{footnotesize}
\parbox{\textwidth}{\emph{Note.} We show the time series of the average daily log-realized variance of ES and TY. The sample period is from March 18, 2010 to October 14, 2021, or $T = 2{,}891$ days in total.}
\end{footnotesize}
\end{center}
\end{figure}

The realized copula of volatility is illustrated in Panel B of Figure \ref{figure:cme-variance}. We plot 500 randomly drawn pairwise observations of $\left( \widehat{F}_{n,T} \big( \log( \hat{V}_{t}^{ \text{ES}}) \big), \widehat{G}_{n,T} \big( \log( \hat{V}_{t}^{ \text{TY}}) \big) \right)$ and add decile contours based on the whole sample. We contrast the latter with those implied by a parametric Gumbel copula, where the single parameter was estimated by maximum likelihood to be $\hat{ \theta}_{ \text{MLE}} = 1.4786$. Overall, there is a close alignment between them.

To provide an alternative graphical representation of the realized copula of volatility, we construct a couple of descriptive statistics that are widely used in the literature. The first is an exceedance measure called tail concentration. It studies the amount of probability mass in the lower-left and upper-right quadrant of a bivariate distribution function, here in the form of a copula \citep[see, e.g.,][]{joe:93a,sibuya:60a}. In particular, the lower (L) and upper (U) tail concentration functions are defined as
\begin{equation} \label{equation:tail-concentration}
{L}(z) = \mathbb{P}(U < z \mid V < z) = \frac{C(z,z)}{z} \quad \text{and} \quad
{U}(z) = \mathbb{P}(U > z \mid V > z) = \frac{1-2z+C(z,z)}{1-z},
\end{equation}
for any $z \in (0,1)$.\footnote{We note that the tail concentration function is symmetric, since its value does not change if we swap the position of $U$ and $V$, because the latter are uniformly distributed.}

\begin{figure}[ht!]
\begin{center}
\caption{Scatter plot of log-realized variance and realized copula of volatility.}
\label{figure:cme-variance}
\begin{tabular}{cc}
\small{Panel A: Scatter plot.} & \small{Panel B: Realized copula.} \\
\includegraphics[height=8cm,width=0.47\textwidth]{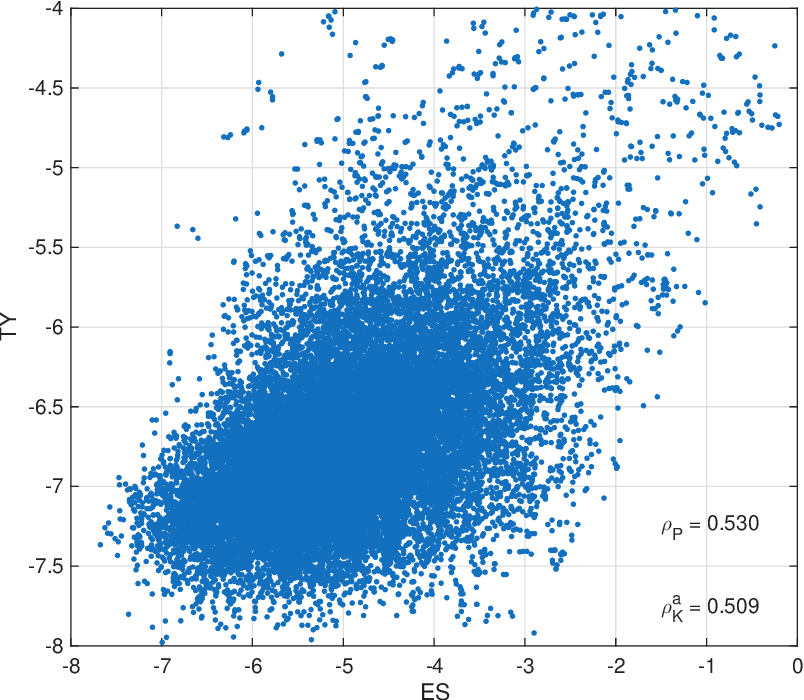} &
\includegraphics[height=8cm,width=0.47\textwidth]{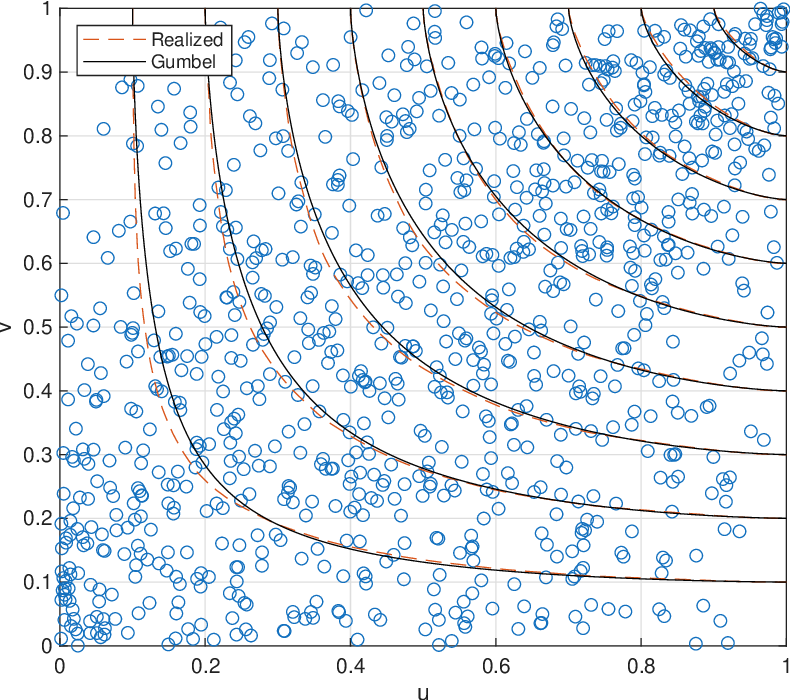}
\end{tabular}
\begin{footnotesize}
\parbox{\textwidth}{\emph{Note.} In Panel A, we show a scatter of log-realized variance of ES plotted against the log-realized variance of TY. In the lower right-hand corner, $\rho_{x}$ is the sample correlation coefficient, where $x = \text{P}$ is Pearson's linear correlation and $x = \text{K}$ is Kendall's tau rank correlation. In Panel B, we draw 500 pairwise observations at random from our nonparametric realized copula of volatility. We superimpose deciles computed over the whole sample, which we contrast with those inferred from the cumulative distribution function of a Gumbel copula (with $\hat{ \theta}_{ \text{MLE}} = 1.4786$).}
\end{footnotesize}
\end{center}
\end{figure}

At one end of the domain, $L$ is convergent toward $\lim_{z \rightarrow 1^{-}}L(z) = 1$, whereas $R$ has the limit $\lim_{z \rightarrow 0^{+}} R(z) = 1$. This holds for every copula, and so it is not very informative. In the other end, however, $L(z)$ and $R(z)$ can converge to any number in the interval $[0,1]$. Provided the limits exist, the lower and upper tail concentration parameters are defined as $\lambda_{L} = \lim_{z \rightarrow 0^{+}} L(z)$ and $\lambda_{U} = \lim_{z \rightarrow 1^{-}} R(z)$. If either limit is nonzero, we say there is asymptotic dependence in that tail. Furthermore, noting that $L(0.5) = R(0.5)$, it is common to plot $T(z) = \min(L(z),R(z))$ as the aggregate tail concentration function.

The second measure is Kendall's distribution function of a copula, which is defined as
\begin{equation} \label{equation:kendall}
K(z) = \mathbb{P}(C(U,V) \leq z).
\end{equation}
Here, $C(U,V)$ is a univariate random variable, and $(U,V)$ are uniform with bivariate distribution function $C(u,v)$. This can be interpreted as the multivariate version of the probability integral transform \citep[see, e.g.,][]{genest-rivest:93a}.\footnote{The connection between this probability measure and Kendall's tau is $\rho_{ \text{K}} = 3-4 \int_{0}^{1}K(z) \mathrm{d}z$ \citep[see, e.g.,][]{schweizer-wolff:81a}.}

\begin{figure}[ht!]
\begin{center}
\caption{Graphical representation of the realized copula of volatility.}
\label{figure:cme-copula}
\begin{tabular}{cc}
\small{Panel A: Tail concentration.} & \small{Panel B: Kendall's distribution function.} \\
\includegraphics[height=8cm,width=0.47\textwidth]{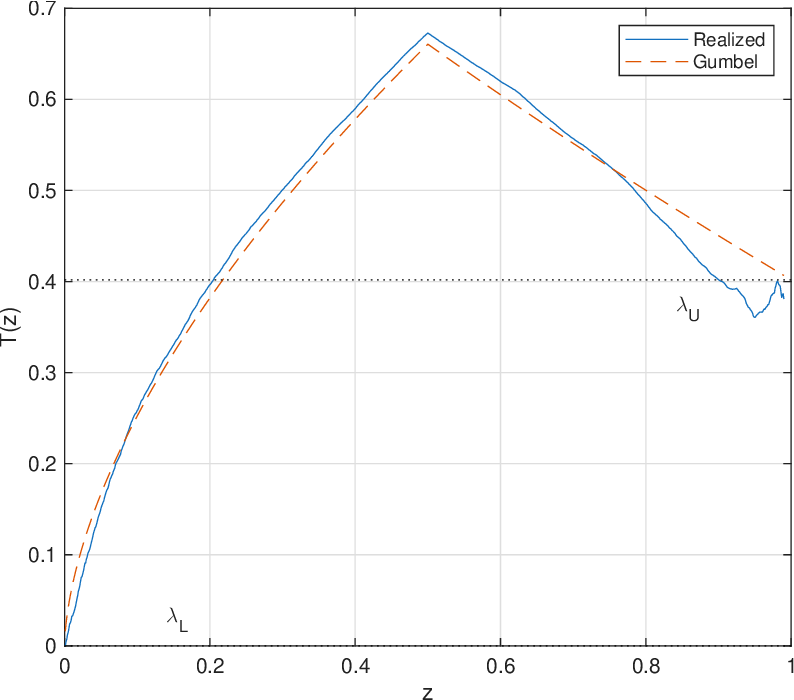} &
\includegraphics[height=8cm,width=0.47\textwidth]{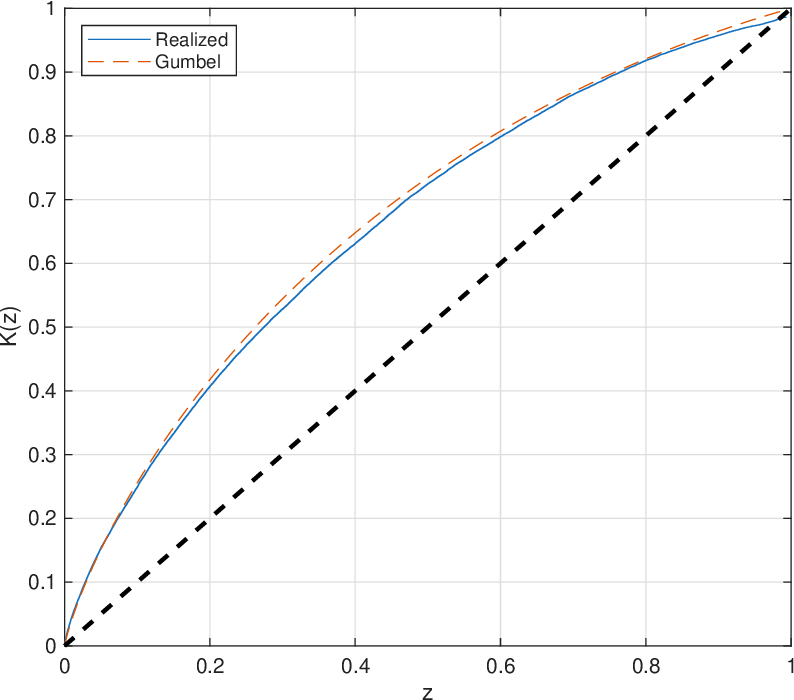}
\end{tabular}
\begin{footnotesize}
\parbox{\textwidth}{\emph{Note.} In Panel A, we show the tail concentration function of the nonparametric realized copula of volatility, which we compare to the one of the parametric Gumbel copula, where the coefficient was estimated by maximum likelihood at $\hat{ \theta}_{ \text{MLE}} = 1.4786$. Also reported (as dashed horizontal lines) are the lower and upper tail concentration parameter, which are $\lambda_{L} = 0$ and $\lambda_{U} = 2 - 2^{1/ \theta}$ for the Gumbel copula. In Panel B, we plot the Kendall's distribution function from \eqref{equation:kendall}. The 45$^{\circ}$-line is a lower bound on $K(z)$ induced by the Fr\'{e}chet-Hoeffding theorem.}
\end{footnotesize}
\end{center}
\end{figure}

We plot the estimated tail concentration function of the realized copula between the ES and TY (log-)variance processes in Panel A of Figure \ref{figure:cme-copula}. In the graph, we also report the tail dependence of the Gumbel copula, again with $\hat{ \theta}_{ \text{MLE}} = 1.4786$. We recall that for the Gumbel copula, $\lambda_{L} = 0$ and $\lambda_{U} = 2 - 2^{1/ \theta}$, which are superimposed as horizontal dotted lines. Overall, there is compelling evidence in favor of asymmetric tail dependence. While the lower tail concentration fizzles out and converges to zero, there is a nontrivial upper tail parameter, which suggests that elevated levels of volatility in the equity and bond markets are likely to occur in parallel. This feature is embedded in the Gumbel copula. Indeed, the tail concentration function of the realized copula of volatility and the associated one of the Gumbel copula are extremely close with a minimal distance, except for a small deviation in the far upper-right tail. We observe that the Gumbel copula has a Kendall’s rank correlation of $\rho_{ \text{K}} = 1 - 1 / \theta$. It thus describes variables that exhibit positive comovement, which is stronger for extremely large values of $\theta$. We estimate an upper tail concentration of about $\hat{ \lambda}_{U} = 0.40$ and a rank correlation of $\hat{ \rho}_{ \text{K}}^{a} = 0.50$.

We corroborate this in Panel B of the figure by showing the estimated Kendall's distribution function. We overlay the population counterpart of the Gumbel copula. The latter is expressible in closed-form, i.e. $K(z) = z - \log(z^{1/ \theta})$. We observe a near-perfect evolution with this measure over the entire support.

To assess goodness-of-fit, we examine a composite null hypothesis with the Gumbel, Clayton, and Independence copulas as contenders for describing the data. The Clayton and Independence copulas are both heavily rejected with a $P\text{-value}$ close to zero, however. Conversely, the Gumbel copula is more appropriate with a $P\text{-value} = 0.1136$.

Taken together, the goodness-of-fit test coupled with the descriptive analysis from Figures \ref{figure:cme-variance} -- \ref{figure:cme-copula} provide overwhelming evidence that the Gumbel copula offers a very good approximation of the codependency between the variance processes of the ES and TY futures contracts.

\section{Conclusion} \label{section:conclusion}

We propose the realized copula of volatility as a nonparametric tool for studying the codependency in the stochastic volatility component of a multivariate continuous-time asset price process. In analogy with the classical copula framework, the aim is to decouple the marginal behavior of the univariate volatility processes from their dependence structure. The distinctive challenge in this setting is that volatility is latent and has to be recovered with a realized measure constructed from discretely observed high-frequency data of the price process. Thus, inference on the realized copula of volatility requires control of measurement error.

We start by proposing a multivariate extension of the realized distribution function of volatility by \citet{li-todorov-tauchen:13a}, from which the realized copula of volatility is derived. Under in-fill asymptotics with a fixed time span, we then show that it affords a consistent estimator of the empirical copula of the latent stochastic volatility (measuring the codependency of volatility on the part of the sample path observed so far). In a double-asymptotic framework with the long-span dimension also going to infinity, our estimator converges to the stationary marginal copula of volatility, provided it exists. We derive a functional central limit theorem for the realized-copula process, which supports uniform inference on the volatility copula. To render the limit theory operational, we design a feasible estimator of the asymptotic covariance function. Based on this, we develop pointwise confidence intervals, uniform confidence bands, and a goodness-of-fit test to evaluate a hypothesis about the marginal copula of volatility.

A simulation study sheds light on the finite sample properties of the realized copula of volatility and the associated inference procedures under realistic sampling schemes. We show that it accurately recovers the empirical copula of volatility for a fixed time span and converges at the expected rate toward the marginal copula of volatility as the long-span dimension increases. The goodness-of-fit test is found to exhibit good size control and demonstrate good power.

In our empirical application, we construct the realized copula of volatility from transaction price data on futures contracts tracking the aggregate U.S. equity index and treasury bond market. The evidence presented points toward a Gumbel copula with asymptotic upper-tail dependence as being highly representative of the behavior of the realized variance processes in these data, in line with the notion that volatility across major asset classes tends to spike in tandem during financial market stress.

The asymptotic theory developed here for the latent volatility of semimartingale processes and the associated inference based on local realized proxies of spot variance can be extended in several directions in future research. First, the baseline analysis can be pursued in the presence of market microstructure noise by building on the univariate treatment of the realized distribution function of volatility in \citet{christensen-thyrsgaard-veliyev:19a}. A second avenue is to allow the realized copula of volatility to vary with the state of the economy or with observable macro-financial indicators, so that dependence in volatility can differ across regimes or along the business cycle. A third extension is to embed the volatility copula in a higher-dimensional framework together with a model for time-varying correlations, thereby delivering a unified copula-based representation of the dynamics of the full covariance matrix. These developments require additional structure but can be built directly on the econometrics and inference procedures developed in this paper.

\pagebreak

\appendix

\section{Appendix} \label{app:proofs}

In this section, we prove the theoretical results presented in the main text. Throughout, we denote by $C$ a generic positive constant. In the fixed $T$ setting, without loss of generality from the standard localization procedure, as stated in Section 4.4.1 of \citet{jacod-protter:12a}, Assumptions \ref{assumption:jump-activity} and \ref{assumption:sigma} can be strengthened into the following stronger statements.

\begin{assuB}
The processes $b^{Z}$ and $\sigma^{Z}$ are bounded. Moreover, for some bounded nonnegative function $\Gamma$ on $\mathbb{R}$, $| \delta^{Z}( \omega, t, z)| \leq \Gamma(z)$ and $\int_{ \mathbb{R}} (1 \wedge \Gamma(z))^{r} \lambda( \mathrm{d}z) < \infty$.
\end{assuB}

\begin{assuBtilde}
The processes $\tilde{b}^{Z}$ and $\tilde{ \sigma}^{Z}$ are bounded. Moreover, for some bounded nonnegative function $\tilde{ \Gamma}$ on $\mathbb{R}$, $| \tilde{ \delta}^{Z}( \omega, t, z)| \leq \tilde{ \Gamma}(z)$ and $\int_{ \mathbb{R}} (1 \wedge \tilde{ \Gamma}(z))^{ \tilde{r}} \tilde{ \lambda}( \mathrm{d}z) < \infty$.
\end{assuBtilde}

The following lemma, which extends Lemma 1 of \citet{li-todorov-tauchen:13a} to the multivariate setting, gives the consistency for the estimation of a general functional of the variance processes.
\begin{lem} \label{lemma:precursor}
Let $g: \mathbb{R}_{+}^{2} \mapsto [0,1]$ be a measurable function and $D_{g}$ be the collection of discontinuity point pairs of $g$. Suppose that
\begin{enumerate}
\item[(i)] Assumption \ref{assumption:jump-activity} holds for $r = 2$.
\item[(ii)] For Lebesgue almost everywhere $s \in [0,T]$, $\mathbb{P} \big((V_{s}^{X}, V_{s}^{Y}) \in D_{g}  \big) = 0$.
\end{enumerate}
Then, as $\Delta_{n} \rightarrow 0$,
\begin{equation*}
\int_{0}^{T}g( \hat{V}_{s}^{X}, \hat{V}_{s}^{Y}) \mathrm{d}s \overset{ \mathbb{P}}{ \longrightarrow} \int_{0}^{T}g(V_{s}^{X}, V_{s}^{Y}) \mathrm{d}s.
\end{equation*}
\end{lem}

\begin{proof}[Proof of Lemma \ref{lemma:precursor}:]
We denote $k_{n} \Delta_{n} = u_{n}$ and set $\hat{V}_{s}^{Z+} = \hat{V}_{ \lceil{ \frac{s}{u_n}} \rceil u_{n}}^{Z+}$, for $Z=X,Y$. By Lemma 1 of \citet{li-todorov-tauchen:13a}, it follows that $\hat{V}_{s}^{Z+} \overset{ \mathbb{P}}{ \longrightarrow} V_{s}$ for all $s \in \big(0, (\lfloor \frac{T}{u_{n}} \rfloor-1)u_{n} \big)$, delivering that $( \hat{V}_{s}^{X+}, \hat{V}_s^{Y+}) \overset{ \mathbb{P}}{ \longrightarrow} (V_{s}^{X}, V_{s}^{Y})$. Thus, by condition (ii) and the Bounded Convergence Theorem (BCT), for Lebesgue almost everywhere $s \in [0,T]$, we get
\begin{equation*}
\mathbb{E} \left[ |g( \hat{V}_{s}^{X+}, \hat{V}_{s}^{Y+}) - g(V_{s}^X, V_{s}^{Y})| \right] = \mathbb{E} \left[|(g( \hat{V}_{s}^{X+}, \hat{V}_{s}^{Y+}) - g(V_{s}^{X}, V_{s}^{Y})) \mathbbm{1}_{ \left \{(V_{s}^{X}, V_{s}^{Y}) \notin D_{g} \right \}}| \right] \rightarrow 0.
\end{equation*}
Furthermore,
\begin{equation*}
\int_{0}^{T}g( \hat{V}_{s}^{X}, \hat{V}_{s}^{Y}) \mathrm{d}s = \int_{0}^{( \lfloor \frac{T}{u_{n}} \rfloor -1)u_{n}}g( \hat{V}_{s}^{X+}, \hat{V}_{s}^{Y+}) \mathrm{d}s + \int_{0}^{u_{n}}g( \hat{V}_{s}^{X}, \hat{V}_{s}^{Y}) \mathrm{d}s + \int_{ \lfloor \frac{T}{u_{n}} \rfloor u_{n}}^{T}g( \hat{V}_{s}^{X}, \hat{V}_{s}^{Y}) \mathrm{d}s.
\end{equation*}
 Since the function $g$ is bounded on $[0,1]$, then
\begin{equation*}
\mathbb{E} \left( \Big| \int_{0}^{T} \big(g( \hat{V}_{s}^{X}, \hat{V}_{s}^{Y}) - g(V_{s}^{X}, V_{s}^{Y}) \big) \mathrm{d}s \Big| \right)
\leq Ku_{n} + \int_{0}^{( \lfloor \frac{T}{u_{n}} \rfloor-1)u_{n}} \mathbb{E} \left( \Big|g( \hat{V}_{s}^{X+}, \hat{V}_{s}^{Y+}) -g(V_{s}^{X}, V_{s}^{Y}) \Big| \right) \mathrm{d}s.
\end{equation*}
To complete the proof, we follow up with a second application of the BCT.
\end{proof}
We now turn our attention to the proof of Theorem \ref{theorem:consistency}, which is established with Lemma \ref{lemma:precursor} as a precursor for the concrete test function $g(\cdot, \cdot \cdot) = \mathbbm{1}_{\{ \cdot \leq x, \cdot \cdot \leq y \}}$.

\begin{proof}[Proof of Theorem \ref{theorem:consistency}:]

This is a bivariate version of Theorem 1 of \citet{li-todorov-tauchen:13a} and Theorem 3.1 of \citet{christensen-thyrsgaard-veliyev:19a}. \\

\noindent Part (1): We fix a point $(x,y) \in \mathbb R_{+}^{2}$. By Assumption \ref{asu4}, $H_{T}$ is almost surely continuous, so
\begin{equation*}
H_{T}(x,y) = H_{T}(x-,y) = H_{T}(x,y-) = H_{T}(x-,y-),
\end{equation*}
almost surely. Hence, we deduce that
\begin{equation*}
\int_{0}^{T} \mathbb{P}(V_{t}^{X} = x, V_{t}^{Y} = y) \mathrm{d}t = T \cdot \mathbb{E} \left[H_{T}(x,y) - H_{T}(x,y-) - H_{T}(x-,y) + H_{T}(x-,y-) \right]
=0.
\end{equation*}
Therefore, we get $\mathbb{P}(V_{t}^{X} = x, V_{t}^{Y} = y) = 0$ for Lebesgue almost everywhere for $t \in [0, T]$, which establishes condition (ii) in Lemma \ref{lemma:precursor}. The pointwise consistency of $\widehat{H}_{n,T}(x,y)$, namely $\widehat{H}_{n, T}(x,y) \overset{ \mathbb{P}}{ \longrightarrow} H_{T}(x,y)$, follows from Theorem 1 of \citet{li-todorov-tauchen:13a} and Lemma \ref{lemma:precursor}. Since $\widehat{H}_{n,T}$ and $H_{T}$ are bounded and non-decreasing, pointwise convergence also implies locally uniform convergence.

Since the paths of the volatility processes $\sigma_{t}^{X}$ and $\sigma_{t}^{Y}$ are c\`{a}dl\`{a}g, for any $\eta>0$, we can find $0< M< \infty$, such that
\begin{equation} \label{conb}
\mathbb{P} \left( \sup_{t \in [0,T]} \{V_{t}^{X} \vee V_{t}^{Y} \} > M \right) < \eta.
\end{equation}
This implies that $\mathbb{P}(H_T(M,M) \neq 1) < \eta$, hence, $\mathbb{P}(F_{T}(M) \neq 1) < \eta$ and $\mathbb{P}(G_{T}(M) \neq 1) < \eta$. Then, for any $\epsilon > 0$, we get the decomposition
\begin{align*}
\mathbb{P} \left( \sup_{(x,y) \in \mathbb{R}_{+}^{2}}| \widehat{H}_{n,T}(x,y) - H_{T}(x,y)| > \epsilon \right) &\leq \mathbb{P} \left( \sup_{(x,y) \in [0, M]^2}| \widehat{H}_{n,T}(x,y) - H_{T}(x,y)| > \epsilon \right) \\
&+ \mathbb{P} \left( \sup_{(x,y) \in (M, \infty) \times [0, M]}| \widehat{H}_{n,T}(x,y) - H_{T}(x,y)| > \epsilon \right) \\
&+ \mathbb{P} \left( \sup_{(x,y) \in [0, M] \times (M, \infty)}| \widehat{H}_{n,T}(x,y) - H_{T}(x,y)| > \epsilon \right) \\
&+ \mathbb{P} \left( \sup_{(x,y) \in (M, \infty)^{2}}| \widehat{H}_{n,T}(x,y) - H_{T}(x,y)| > \epsilon \right) \\
&\equiv \mathrm{I}_{n,T} + \mathrm{II}_{n,T} + \mathrm{III}_{n,T} + \mathrm{IV}_{n,T}.
\end{align*}
The locally uniform convergence of $H_{n,T}(x,y)$ implies that $\limsup_{n \rightarrow \infty} \mathrm{I}_{n,T} = 0$, which is therefore asymptotically negligible. Second,
\begin{align*}
\mathbb{P} \left( \sup_{x > M, y \leq M}| \widehat{H}_{n,T}(x,y) - H_{T}(x,y)| > \epsilon \right) &\leq \mathbb{P} \left( \sup_{x > M, y \leq M}| \widehat{H}_{n,T}(x,y ) - \widehat{H}_{n,T}(M,y)| > \epsilon/3 \right) \\
&+ \mathbb{P} \left( \sup_{y \leq M}| \widehat{H}_{n,T}(M,y) - H_{T}(M,y)| > \epsilon/3 \right) \\
&+ \mathbb{P} \left( \sup_{x> M, y \leq M}|H_{T}(M,y) - H_{T}(x,y)| > \epsilon/3 \right)\\
&\leq \mathbb{P} \left( \sup_{y \leq M}| \widehat{H}_{n,T}( \infty, y) - \widehat{H}_{n,T}(M,y)| > \epsilon/3 \right) \\
&+ \mathbb{P} \left( \sup_{y \leq M}| \widehat{H}_{n,T}(M,y) - H_{T}(M,y)| > \epsilon/3 \right) \\
&+\mathbb{P} \left( \sup_{y \leq M}|H_{T}(M,y) - H_{T}( \infty, y)| > \epsilon/3 \right).
\end{align*}
Now, by \eqref{conb} and locally uniform convergence, we obtain
\begin{align*}
\limsup_{n \rightarrow \infty} \mathbb{P} \left( \sup_{y \in [0,M]}| \widehat{H}_{n,T}( \infty,y) - \widehat{H}_{n,T}(M,y)| > \epsilon/3 \right) &\leq \limsup_{n \rightarrow \infty} \mathbb{P} \left(|1- \widehat{F}_{n,T}(M)| > \epsilon/3 \right) \\
&\leq \limsup_{n \rightarrow \infty} \mathbb{P} \left(|1-F_{T}(M)| > \epsilon/6 \right) \\
&+ \limsup_{n \rightarrow \infty} \mathbb{P} \left(| \widehat{F}_{n,T} - F_{T}(M)| > \epsilon/6 \right) \\
&< \eta.
\end{align*}
We also get that
\begin{equation*}
\limsup_{n \rightarrow \infty} \mathbb{P} \left( \sup_{y \leq M}| \widehat{H}_{n,T}(M,y) - H_{T}(M,y)| > \epsilon/3 \right) = 0,
\end{equation*}
and
\begin{equation*}
\mathbb{P} \left( \sup_{y \leq M}|H_{T}(M,y) - H_{T}( \infty, y)| > \epsilon/3) \leq \mathbb{P}(|1-F_{T}(M)| > \epsilon/3 \right) < \eta.
\end{equation*}
Taking $\eta$ arbitrarily small, we conclude that $\limsup_{n \rightarrow \infty} \mathrm {II}_{n,T} = 0$ and $\limsup_{n \rightarrow \infty} \mathrm{III}_{n,T} = 0$.

Finally, following the above line of reasoning,
\begin{align*}
\limsup_{n \rightarrow \infty}  \mathbb{P} \left( \sup_{(x,y) \in [M, \infty)^{2}}| \widehat{H}_{n,T}(x,y) - H_{T}(x,y)| > \epsilon \right) &\leq \limsup_{n \rightarrow \infty} \mathbb{ P} \left(|1- \widehat{H}_{n,T}(M, M)| > \epsilon/3 \right) \\
&+ \limsup_{n \rightarrow \infty} \mathbb{P} \left(| \widehat{H}_{n,T}(M,M) - H_{T}(M,M)| > \epsilon/3 \right) \\
&+ \limsup_{n \rightarrow \infty} \mathbb{P} \left(|1-H_{T}(M,M)| > \epsilon/3 \right) \\
&\leq 2 \eta.
\end{align*}
Again, we can choose $\eta$ as arbitrarily small as needed, so it follows that $\limsup_{n \rightarrow \infty} \mathrm{IV}_{n,T} = 0$. Hence, the proof of part (1) is done.

Part (2): By Corollary 1 in \citet{li-todorov-tauchen:13a}, which builds on Lemma 21.2 in \citet{vaart:98a}, we deduce that
\begin{equation*}
\big( \widehat{Q}_{n,T}^{X}(u), \widehat{Q}_{n,T}^{Y}(v) \big) \overset{ \mathbb{P}}{ \longrightarrow} \big(Q_{T}^{X}(u), Q_{T}^{Y}(v) \big).
\end{equation*}
Note that
\begin{align*}
| \widehat{C}_{n,T}(u,v) - C_{T}(u,v)|&\leq | \widehat{H}_{n,T}(Q_{n,T}^{X}(u), Q_{n,T}^{Y}(v)) - H_{T}(Q_{n,T}^{X}(u), Q_{n,T}^{Y}(v))| \\
&+ |H_{T}(Q_{n,T}^{X}(u), Q_{n,T}^{Y}(v)) - H_{T}(Q_{T}^{X}(u), Q_{T}^{Y}(v))|,
\end{align*}
so pointwise convergence of $C_{n,T}(u,v)$ follows from the convergence of $\widehat{H}_{n,T}(x,y)$ and pointwise continuity of $Q_{T}^{X}(u)$ and $Q_{T}^{Y}(v)$. The uniform convergence follows from part (1) and the uniform continuity of $H_{T}(x,y)$.
\end{proof}

The proof of Theorem \ref{thm2} relies on the upcoming lemma, which shows that the estimation error of the realized distribution function of volatility and the realized copula of volatility, centered around their empirical counterpart and scaled with $\sqrt{T}$, converges to zero in probability as $\Delta_{n} \rightarrow 0$ and $T \rightarrow \infty$. It extends Lemma 2 and Theorem 2 of \citet{li-todorov-tauchen:13a} for $T$ fixed.

\begin{lem} \label{lema2}
We suppose that Assumptions \ref{assumption:jump-activity} - \ref{asu4} hold and that $k_{n} = \Delta_{n}^{- \gamma}$. We choose $\frac{r-1}{r} < \varpi < \frac{1}{2}$ and $r(1- \varpi) < \gamma < 1$ for $r>1$ or $0 < \varpi < \frac{1}{2}$ and $1-(2-r) \varpi < \gamma < 1$ for $r\leq 1$. Then, under either of the following conditions:
\begin{enumerate}
\item $V_{t}^{X}$ and $V_{t}^{Y}$ are continuous on $\mathbb{R}_{+}$, and:
\begin{itemize}
\item[1).] $r\leq 1$: $\sqrt{T} \left( \Delta_{n}^{ \gamma-1+(2-r) \varpi} + \Delta_{n}^{ \gamma/2 - \iota} + \Delta_{n}^{(1- \gamma)/2 - \iota} \right) \rightarrow 0$,
\item[2).] $r > 1$: $\sqrt{T} \left( \Delta_{n}^{ \gamma/r-(1- \varpi) - \iota} + \Delta_{n}^{ \gamma/2 - \iota} + \Delta_{n}^{(1- \gamma)/2- \iota} \right) \rightarrow 0$,
\end{itemize}
\item $V_{t}^{X}$ or $V_{t}^{Y}$ are discontinuous on $\mathbb{R}_{+}$, and:
\begin{itemize}
\item $r \leq 1$: $\sqrt{T} \left( \Delta_{n}^{ \gamma-1+(2-r) \varpi} + \Delta_{n}^{ \gamma/2-\iota} + \Delta_{n}^{(1- \gamma)/2- \iota} + \Delta_{n}^{ \frac{1- \gamma}{1+ \tilde{r}}- \iota} \right) \rightarrow 0$,
\item $r > 1$: $\sqrt{T} \left( \Delta_{n}^{ \gamma/r-(1- \varpi)- \iota} + \Delta_{n}^{ \gamma/2-\iota} + \Delta_{n}^{(1- \gamma)/2- \iota} + \Delta_{n}^{ \frac{1-\gamma}{1+ \tilde{r}}- \iota} \right) \rightarrow 0$.
\end{itemize}
\end{enumerate}
Then, it follows as $\Delta_{n} \rightarrow 0$ and $T \rightarrow \infty$ that
\begin{equation*}
\sqrt{T} \left( \widehat{H}_{n,T}(x,y) - H_{T}(x,y) \right) \overset{ \mathbb{P}}{ \longrightarrow} 0 \qquad \text{and} \qquad \sqrt{T} \left( \widehat{C}_{n,T}(u,v) - C_T(u,v) \right) \overset{ \mathbb{P}}{ \longrightarrow} 0.
\end{equation*}
\end{lem}

\begin{proof}[Proof of Lemma \ref{lema2}:] 1) If the volatility processes $V^{X}$ and $V^{Y}$ are continuous, we define
\begin{align*}
&\eta_{n}^{X} = \sup_{t \in[0, \infty)}| \hat{V}_{t}^{X} - V_{t}^{X}|, \quad   \eta_{n}^{Y} = \sup_{t \in[0, \infty)}| \hat{V}_{t}^{Y} - V_{t}^{Y}|, \\
&\beta_{n}^{X} = \sup_{x \in[0, \infty)}| \widehat{F}_{n,T}(x) - F_{T}(x)|, \quad \beta_{n}^{Y} = \sup_{y \in[0, \infty)}| \widehat{G}_{n,T}(y) - G_T(y)|.
\end{align*}
Following Lemma 2 and Theorem 2 of \citet{li-todorov-tauchen:13a}, we know that $\eta_{n}^{X}$ and  $\beta_{n}^{X}$ are both of order $O_{P}(d_{n})$, with
 \begin{align*}
 d_{n} = \left\{
\begin{array}{ll}
\Delta_{n}^{ \gamma-1 + (2-r) \varpi} + \Delta_{n}^{ \gamma/2- \iota} + \Delta_{n}^{(1-\gamma)/2- \iota}, &r \leq 1, \\[0.25cm]
\Delta_{n}^{ \gamma/r-(1- \varpi)- \iota} + \Delta_{n}^{ \gamma/2- \iota} + \Delta_{n}^{(1- \gamma)/2- \iota}, & r>1.
\end{array}
\right.
\end{align*}
Now, let $\alpha_{n}^{X} \equiv \sup_{t \in[0, \infty)}| \widehat{F}_{n,T}( \hat{V}_{t}^{X}) - F_{T}(V_{t}^{X})|$ and $\alpha_{n}^{Y} \equiv \sup_{t \in[0, \infty)}| \widehat{G}_{n,T}( \hat{V}_{t}^{Y}) - G_{T}(V_{t}^{Y})|$. Then, \begin{align*}
\left| \widehat{C}_{n,T}(u,v) - C_{T}(u,v) \right| &= \frac{1}{T} \int_{0}^{T} \left| \mathbbm{1}_{ \left\{ \widehat{F}_{n,T}( \hat{V}_{t}^{X}) \leq u, \widehat{G}_{n,T}( \hat{V}_{t}^{Y}) \leq v \right\}} - \mathbbm{1}_{ \left\{ F_{T}({V}_{t}^{X}) \leq u, G_{T}({V}_{t}^{Y}) \leq v \right\}} \right| \mathrm{d}t \\
&\leq \frac{1}{T} \int_{0}^{T} \left| \mathbbm{1}_{ \{ \widehat{F}_{n,T}( \hat{V}_{t}^{X}) \leq u \}} - \mathbbm{1}_{ \{F_{T}({V}_{t}^{X}) \leq u \}} \right| \mathbbm{1}_{ \{ \widehat{G}_{n,T}( \hat{V}_{t}^{Y}) \leq v\}} \mathrm{d}t \\
&+ \frac{1}{T} \int_{0}^{T} \left| \mathbbm{1}_{ \{ \widehat{G}_{n,T}( \hat{V}_{t}^{Y}) \leq v \}} - \mathbbm{1}_{ \{G_{T}({V}_{t}^{Y}) \leq v \}} \right| \mathbbm{1}_{ \{ \widehat{F}_{n,T}({V}_{t}^{X}) \leq u \}} \mathrm{d}t \\
&\leq \frac{1}{T} \int_{0}^{T} \mathbbm{1}_{ \{u - \alpha_{n}^{X} \leq F_{T}({V}_{t}^{X}) \leq u + \alpha^{X}_{n} \}} \mathrm{d}t + \frac{1}{T} \int_{0}^{T} \mathbbm{1}_{ \{v - \alpha_{n}^{Y} \leq G_{T}({V}_{t}^{Y}) \leq v + \alpha_{n}^{Y} \}} \mathrm{d}t.
\end{align*}
Notice that
\begin{equation*}
\frac{1}{T} \int_{0}^{T} \mathbbm{1}_{ \{F_{T}({V}_{t}^{X}) \leq u \}} \mathrm{d}t = \frac{1}{T} \int_{0}^{T} \mathbbm{1}_{ \{V_{t}^{X} \leq Q_{T}^{X}(u) \}} \mathrm{d}t = F_{T}(Q_{T}^{X}(u)) = u.
\end{equation*}
Hence,
\begin{equation*}
\sqrt{T} \left| \widehat{C}_{n,T}(u,v) - C_{T}(u,v) \right| \leq K \sqrt{T}( \alpha_{n}^{X} + \alpha_{n}^{Y}).
\end{equation*}
To calculate the order of $\alpha_{n}^{X}$, we observe that
\begin{align*}
\alpha_{n}^{X} &\leq \sup_{t \in[0, \infty)}| \widehat{F}_{n,T}( \hat{V}_{t}^{X}) - F_{T}( \hat{V}_{t}^{X})| +
\sup_{t \in[0, \infty)}|F_{T}( \hat{V}_{t}^{X}) - F_{T}(V_{t}^{X})| \\
&\leq \beta_{n}^{X} + \sup_{t \in[0, \infty)}|F_{T}( \hat{V}_{t}^{X}) - F_{T}(V_{t}^{X})|.
\end{align*}
To handle the second term in the inequality, namely $\sup_{t \in[0, \infty)}|F_{T}( \hat{V}_{t}^{X}) - F_{T}(V_{t}^{X})|$, we follow the proof of Lemma 2(b) in \citet{li-todorov-tauchen:13a} by observing that
\begin{equation} \label{equation:OrderOfFV}
\sup_{t \in[0, \infty)}|F_{T}( \hat{V}_{t}^{X}) - F_{T}(V_{t}^{X})| \leq \int_{ \min \{V_{t}^{X}, \hat{V}_{t}^{X} \}}^{ \max \{V_{t}^{X}, \hat{V}_{t}^{X} \}} f_{T}(z) \mathrm{d}z = O_{P}(d_{n}).
\end{equation}
Similarly, we can get that
\begin{eqnarray*}
\alpha_{n}^{Y} = O_P(d_{n}).
\end{eqnarray*}
Therefore, if the conditions in 1) hold, we obtain
\begin{equation*}
\sqrt{T}( \alpha_{n}^{X} + \alpha_{n}^{Y}) = o_{P}(1).
\end{equation*}
2) To deal with the situation, where the volatility processes can jump, we introduce the set
\begin{equation*}
\mathcal{I}_{n}( \epsilon) = \left\{0 \leq i \leq \lfloor T/u_{n} \rfloor-1 : \mu \left(I_{n,i} \times \left \{z: \tilde{ \Gamma}(z) > \epsilon \right \} \right) = 0 \right \}.
\end{equation*}
where $I_{n,i} = [iu_{n}, (i+1)u_{n})$. We also denote $B_{n}( \epsilon) = \cup_{i \in \mathcal{I}_{n}( \epsilon)} I_{n,i}$ and $B_{n}^{c}( \epsilon) = [0,T] \setminus B_{n}( \epsilon)$. Here, $\mathcal{I}_{n}( \epsilon)$ is a set of time indices that exclude the locations of the big jumps in volatility of size larger than $\epsilon$. In this case,
\begin{equation*}
\left| \widehat{C}_{n,T}(u,v) - C_{T}(u,v) \right| \leq \frac{1}{T} \sum_{k=1}^{ \lfloor T \rfloor} \int_{t_{k-1}}^{t_k} \left| \mathbbm{1}_{ \{ \widehat{F}_{n,T}( \hat{V}_{t}^{X}) \leq u \}} - \mathbbm{1}_{ \{F_{T}({V}_{t}^{X}) \leq u \}} \right| \mathrm{d}t + \frac{1}{T}.
\end{equation*}
where $t_{k} = { \lfloor k/u_{n} \rfloor}u_{n}$. For every $k$, it holds that
\begin{align*}
\int_{t_{k-1}}^{t_k} \left| \mathbbm{1}_{ \{ \widehat{F}_{n,T}( \hat{V}_{t}^{X}) \leq u \}} - \mathbbm{1}_{ \{F_{T}({V}_{t}^{X}) \leq u \}} \right| \mathrm{d}t &= \int_{[t_{k-1}, t_{k}] \cap B_{n}( \epsilon)} \left| \mathbbm{1}_{ \{ \widehat{F}_{n,T}( \hat{V}_{t}^{X}) \leq u \}} - \mathbbm{1}_{ \{F_{T}({V}_{t}^{X}) \leq u \}} \right| \mathrm{d}t \\
&+ \int_{[t_{k-1}, t_{k}] \cap B_{n}^{c}( \epsilon)} \left| \mathbbm{1}_{ \{ \widehat{F}_{n,T}( \hat{V}_{t}^{X}) \leq u \}} - \mathbbm{1}_{ \{F_{T}({V}_{t}^{X}) \leq u \}} \right| \mathrm{d}t.
\end{align*}
Reusing \citet[][Lemma 5]{li-todorov-tauchen:13a}, we deduce that
\begin{equation*}
\int_{[t_{k-1}, t_{k}] \cap B_{n}( \epsilon)} \left| \mathbbm{1}_{ \{ \widehat{F}_{n,T}( \hat{V}_{t}^{X}) \leq u \}} - \mathbbm{1}_{ \{F_{T}({V}_{t}^{X}) \leq u \}} \right| \mathrm{d}t = O_{P}(d_{n}),
\end{equation*}
and
\begin{equation*}
\int_{[t_{k-1}, t_{k}] \cap B_{n}^{c}( \epsilon)} \left| \mathbbm{1}_{ \{ \widehat{F}_{n,T}( \hat{V}_{t}^{X}) \leq u \}} - \mathbbm{1}_{ \{F_{T}({V}_{t}^{X}) \leq u \}} \right| \mathrm{d}t = O_{P}(d_{n}').
\end{equation*}
Here, $d_{n}' = \Delta_{n}^{(1- \gamma)/(1+ \tilde{r}) - \theta}$ and both equalities hold uniformly in $k$.

Therefore, if the conditions in 2) hold,
\begin{equation*}
\sqrt{T} \left( \widehat{C}_{n,T}(u,v) - C_{T}(u,v) \right) \overset{ \mathbb{P}}{ \longrightarrow} 0.
\end{equation*}
The proof that $\sqrt{T}( \widehat{H}_{n,T}(x,y) - H_T(x,y)) \overset{ \mathbb{P}}{ \longrightarrow} 0$ follows this route too.
\end{proof}

Now, we deliver the proof of Theorem \ref{thm2}. This is based on an application of the central limit theorem for mixing stationary sequences of bounded random variables (see Theorem 18.5.4 in \cite{ibragimov:75a}, Theorem 5.2 in \citet{dehay:05a}, and Lemma \ref{lema2}).

\begin{proof}[Proof of Theorem \ref{thm2}:] We define the empirical process:
\begin{equation*}
 G_{T} \equiv \left \{ \sqrt{T}(H_{T}(x,y) - H(x,y)), (x,y) \in \mathbb{R}_{+}^{2} \right\}.
\end{equation*}
To prove the weak convergence $G_{T}(x,y) \Rightarrow \mathcal{G}(x,y)$, we need to verify (1) convergence of finite-dimensional distributions and (2) tightness. To this end, we write
\begin{equation*}
H_{T}(x,y) - H(x,y) = \frac{1}{T} \sum_{k=1}^{ \lfloor T \rfloor} Z_{k}(x,y) + \frac{1}{T} \int_{ \lfloor T \rfloor}^{T} \mathbbm{1}_{ \left \{V_{t}^{X} \leq x, V_{t}^{Y} \leq y \right\}} \mathrm{d}t,
\end{equation*}
where $Z_{k}(x,y) = \int_{k-1}^{k} \mathbbm{1}_{ \left \{V_{t}^{X} \leq x, V_{t}^{Y} \leq y \right \}} \mathrm{d}t - H(x,y)$. The second term on the right side of the equality is a remainder term, for which we immediately conclude that
\begin{equation*}
\sqrt{T} \frac{1}{T} \int_{ \lfloor T \rfloor}^{T} \mathbbm{1}_{ \left \{V_{t}^{X} \leq x, V_{t}^{Y} \leq y \right\}} \mathrm{d}t \leq \frac{1}{ \sqrt{T}} \rightarrow 0,
\end{equation*}
almost surely. Thus, it suffices to show that
\begin{equation*}
\zeta_{m}(x,y) \equiv \frac{1}{ \sqrt{m}} \sum_{k=1}^{m} Z_{k}(x,y) \Rightarrow  \mathcal{G},
\end{equation*}
where $m = \lfloor T \rfloor$.

(1) \textit{Finite-dimensional convergence in law}: First, for a fixed coordinate $(x,y) \in \mathbb{R}_{+}^{2}$, we note that $Z_{k}(x,y)$ is a sequence of bounded random variables. So, letting $\alpha(n)$ be the strong mixing coefficient of the sequence $(Z_{k})_{k=1}^{ \infty}$, we deduce that $\alpha(n) \leq \alpha_{n-1}$. Hence, by Assumption \ref{asu5}
\begin{equation*}
\sum_{n=1}^{ \infty} \alpha(n) \leq \sum_{t=0}^{ \infty} \alpha_{t} \leq C \sum_{t=0}^{ \infty} t^{-(1+ \gamma)} < \infty.
\end{equation*}
Then, we can deduce by Theorem 18.5.4 of \cite{ibragimov:75a} that
\begin{eqnarray*}
\zeta_{m}(x,y) \overset{d}{ \longrightarrow} N \left(0, \sigma^{2}(x,y) \right).
\end{eqnarray*}
Here, the long-run variance is given by
\begin{eqnarray*}
\sigma^{2}(x,y) = 2 \int_{0}^{ \infty}(H_{t}(x,y) - H(x,y)^{2}) \mathrm{d}t,
\end{eqnarray*}
with $H_{t}(x,y) = \mathbb{P}(V_{0}^{X} \leq x, V_{0}^{Y} \leq y, V_{t}^{X} \leq x, V_{t}^{Y} \leq y)$.

More generally, for any finite collection $\{(x_{i}, y_{i}) \}_{i=1}^{r}$ with $(x_{i},y_{i}) \in \mathbb{R}_{+}^{2}$ and any $v = (v_{1}, \dots, v_{r})^{ \top} \in \mathbb{R}^{r}$, we aim to explain the behavior of the linear combination
\begin{equation*}
L_{m} = \sum_{i=1}^{r} v_{i} \zeta_{m}(x_{i},y_{i}) = \frac{1}{ \sqrt{m}} \sum_{k=1}^{m} D_{k}
\end{equation*}
where
\begin{equation*}
D_{k} = \sum_{i=1}^{r} v_{i}Z_{k}(x_{i},y_{i}).
\end{equation*}
Next, $\{D_{k} \}$ is again a bounded, strong mixing sequence with mixing coefficient
\begin{equation*}
\sum_{h=1}^{ \infty} \alpha_{D}(h) \leq r \sum_{h=1}^{ \infty} \alpha(h) < \infty,
\end{equation*}
so by Theorem 18.5.4 in \cite{ibragimov:75a}, we once again conclude that
\begin{equation*}
L_{m} \overset{d}{ \longrightarrow} N(0, \sigma_{L}^{2}),
\end{equation*}
with
\begin{equation*}
\sigma_{L}^{2} = \mathrm{var}(D_{0}) + 2 \sum_{h=1}^{ \infty} \mathrm{cov}(D_{0}, D_{h}) = v^{ \top} \Sigma_{r} v,
\end{equation*}
in which $\Sigma_{r} = ( \Sigma_{ij})_{1 \leq i,j \leq r}$, and
\begin{equation*}
\Sigma_{ij} = \lim_{m \rightarrow \infty} \mathrm{cov} \left( \zeta_{m}(x_{i},y_{i}), \zeta_{m}(x_{j},y_{j}) \right) = 2 \int_{0}^{ \infty} \big( H_{t}(x_{i},y_{i},x_{j},y_{j}) - H(x_{i},y_{i})H(x_{j},y_{j}) \big) \mathrm{d}t,
\end{equation*}
since the covariance between $\zeta_{m}(x_{i},y_{i})$ and $\zeta_{m}(x_{j},y_{j})$ is given by
\begin{align*}
\mathrm{cov} \left( \zeta_{m}(x_{i},y_{i}), \zeta_{m}(x_{j},y_{j}) \right) &= \frac{1}{m} \int_{0}^{m} \int_{0}^{m} \Big( H_{s,t}(x_{i},y_{i},x_{j},y_{j}) - H(x_{i},y_{j}) H(x_{j},y_{j}) \Big) \mathrm{d}t \mathrm{d}s \\
&= \frac{1}{m} \int_{0}^{m} \int_{0}^{m} \big(H_{|t-s|}(x_{i},y_{i},x_{j},y_{j}) - H(x_{i},y_{i}) H(x_{j},y_{j}) \big) \mathrm{d}t \mathrm{d}s \\
&= \frac{2}{m} \int_{0}^{m} \int_{s}^{m} \big(H_{t-s}(x_{i},y_{i},x_{j},y_{j}) - H(x_{i},y_{i}) H(x_{j},y_{j}) \big) \mathrm{d}t \mathrm{d}s \\
&= 2 \int_{0}^{m} \bigg(1- \frac{t}{m} \bigg) \big(H_{t}(x_{i},y_{i},x_{j},y_{j}) - H(x_{i},y_{i}) H(x_{j},y_{j})\big) \mathrm{d}t.
\end{align*}
Here, $H_{s,t}(x_{i},y_{i},x_{j},y_{j}) = \mathbb{P} \big(V_{s}^{X} \leq x_{i}, V_{s}^{Y} \leq y_{i}, V_{t}^{X} \leq x_{j}, V_{t}^{Y} \leq y_{j} \big)$, and the stationary condition in Assumption \ref{asu5} is applied for the second equality in the display.

As $v$ was arbitrary, the Cram\'{e}r-Wold device implies that
\begin{eqnarray*}
\bigl( \zeta_{m}(x_{1},y_{1}), \dots, \zeta_{m}(x_{r},y_{r}) \bigr) \overset{d}{ \longrightarrow}  N_{r} \bigl( \mathbf{0}, \Sigma_{r} \bigr).
\end{eqnarray*}

(2) \textit{Tightness}: To verify tightness of the sequence of processes $\{ \zeta_{m}(x,y) \}$, we extend the arguments presented in the proof of Theorem 5.2 from \citet{dehay:05a}. In particular, we deduce that
\begin{equation*}
\limsup_{m \rightarrow \infty} \mathbb{P}^{*} \Bigg( \sup_{ \rho \left((x_{i},y_{i}), (x_{j},y_{j}) \right) < \delta}| \zeta_{m}(x_{i},y_{i}) - \zeta_{m}(x_{j},y_{j})| > \eta \Bigg) < \epsilon,
\end{equation*}
for all $\eta > 0$ and $\epsilon > 0$. Here, $\rho \left((x_{i},y_{i}), (x_{j},y_{j}) \right) = |H(x_{i},y_{i}) - H(x_{j},y_{j})|$ and $\mathbb{P}^{*}$ is the outer probability measure associated with $\mathbb{P}$. Hence, $G_{T}(x,y) \Rightarrow \mathcal{G}(x,y)$.

Next, we show that
\begin{equation*}
\sqrt{T} \left(C_{n,T}(u,v) - C(u,v) \right) \Rightarrow \mathcal{C}(u,v).
\end{equation*}
Consider the mapping $\phi(H)(u,v) = H(Q^{X}(u), Q^{Y}(v))$. By \citet[][Lemma 2]{fermanian-radulovic-wegkamp:04a}, the mapping $\phi$ is Hadamard differentiable at the point $H^{ \ast} (\cdot, \cdot \cdot) = H(Q^{X}( \cdot), Q^{Y}( \cdot \cdot))$. Moreover, the Hadamard derivative of $\phi$ at this point is given by
\begin{align*}
\phi_{H^{ \ast}}^{ \prime}(h)(u,v) = h \big(Q^{X}(u), Q^{Y}(v) \big) &- \frac{ \partial H}{ \partial x}(Q^{X}(u), Q^{Y}(v)) \cdot \frac{h(Q^{X}(u), \infty)}{f(Q^{X}(u))} \\
& - \frac{ \partial H}{ \partial y}(Q^{X}(u), Q^{Y}(v)) \cdot \frac{h( \infty, Q^Y(v))}{g(Q^Y(v))}.
\end{align*}
By the previous result, namely $G_{T} \Rightarrow \mathcal{G}$, and the functional delta method, we obtain
\begin{equation*}
\sqrt{T} \left(C_{T}(x,y) - C(x,y) \right) \Rightarrow \phi^{ \prime}_{H^{ \ast}}( \mathcal{G}),
\end{equation*}
where the limit process can be written as
\begin{equation*}
\phi^{ \prime}_{H^{ \ast}}( \mathcal{G}) = \mathcal{G}(Q^{X}(u), Q^{Y}(v)) - \frac{ \partial C}{ \partial u}(u,v) \mathcal{G}(Q^{X}(u), \infty) -\frac{ \partial C}{ \partial v}(u,v) \mathcal{G}( \infty, Q^{Y}(v)),
\end{equation*}
with
\begin{align*}
\frac{ \partial C}{ \partial u}(u,v) &= \frac{ \partial H(Q^{X}(u), Q^{Y}(v))}{ \partial x} \cdot \frac{1}{f(Q^{X}(u))}, \\
\frac{ \partial C}{ \partial v}(u,v) &= \frac{ \partial H(Q^{X}(u), Q^{Y}(v))}{ \partial y} \cdot \frac{1}{g(Q^{Y}(v))}.
\end{align*}
In view of Lemma \ref{lema2}, the proof is complete.
\end{proof}

\begin{proof}[Proof of Corollary \ref{cor1}:] The proof builds on Theorem 4.2 in \citet{dehay:05a} and Lemma \ref{lema2} above. It is enough to prove that
\begin{align*}
&\int_{0}^{T^{ \xi}} \left\{\widehat{H}_{t,n,T}(x,y,x',y') - H_{t,T}(x,y,x',y') \right\} \mathrm{d}t \overset{ \mathbb{P}}{ \longrightarrow} 0, \\
&\int_{0}^{T^{ \xi}} \left\{ \widehat{H}_{n,T}(x,y) \widehat{H}_{n,T}(x',y') - H_{T}(x,y)H_{T}(x',y') \right\} \mathrm{d}t \overset{ \mathbb{P}}{ \longrightarrow} 0, \\
&2\int_{0}^{T^{ \xi}} \left\{H_{t,T}(x,y,x',y') - H_{T}(x,y)H_{T}(x',y') \right\} \mathrm{d}t - C_{ \mathcal{G}}(x,y,x',y') \overset{ \mathbb{P}}{ \longrightarrow} 0.
\end{align*}
Here, $H_{t,T}(x,y,x',y') = \frac{1}{T-T^{ \xi}} \int_{0}^{T-T^{ \xi}} 1_{ \left\{{V}_{s}^{X} \leq x, {V}_{s}^{Y} \leq y,{V}_{s+t}^{X} \leq x', {V}_{s+t}^{Y} \leq y' \right\}} \mathrm{d}s$.

The convergence of the third term follows directly from the proof of Proposition 4.2 in \citet{dehay:05a}.

To establish convergence of the first term, we observe that
\begin{align*}
\left| \widehat{H}_{t,n,T}(x,y,x',y') - H_{t,T}(x,y,x',y') \right| &\leq \frac{1}{T-T^{ \xi}} \int_{0}^{T-T^{ \xi}} \big|\mathbbm{1}_{ \left\{ \hat{V}_{s}^{X} \leq x, \hat{V}_{s}^{Y} \leq y \right\}} - \mathbbm{1}_{ \left\{{V}_{s}^{X} \leq x, {V}_{s}^{Y} \leq y \right\}} \big| \mathrm{d}s \\
&\quad+ \frac{1}{T-T^{ \xi}} \int_{t}^{T-T^{ \xi}+t} \big|\mathbbm{1}_{ \left\{ \hat{V}_{s}^{X} \leq x', \hat{V}_{s}^{Y} \leq y' \right\}} - \mathbbm{1}_{ \left\{{V}_{s}^{X} \leq x', {V}_{s}^{Y} \leq y' \right\}} \big| \mathrm{d}s.
\end{align*}
Here, we apply the inequality $|ab - cd| \leq |a-c| + |b-d|$, for $a,b,c,d \in [0,1]$, and the change of variable $r = s+t$.

Exploiting the arguments presented in the proof for $\widehat H_{n,T}(x,y)-H_T(x,y)$ in Lemma \ref{lema2}, we can show that
\begin{equation*}
\sup_{t \in[0,T^{ \xi}]} \left|\widehat{H}_{t,n,T}(x,y,x',y') - H_{t,T}(x,y,x',y') \right| = o_{P} \left( \frac{1}{ \sqrt{T}} \right).
\end{equation*}
Hence,
\begin{equation*}
\left| \int_{0}^{T^{ \xi}} \left\{\widehat{H}_{t,n,T}(x,y,x',y') - H_{t,T}(x,y,x',y') \right\} \mathrm{d}t \right| = o_{P} \left(T^{ \xi - \frac{1}{2}} \right).
\end{equation*}
Finally, we deal with the asymptotic negligibility of the second term. By Lemma \ref{lema2}, we get
\begin{align*}
&\left| \int_{0}^{T^{ \xi}} \left\{ \widehat{H}_{n,T}(x,y) \widehat{H}_{n,T}(x',y') - H_{T}(x,y)H_{T}(x',y') \right\} \mathrm{d}t \right| \\
&\leq T^{ \xi} \left( \left| \widehat{H}_{n,T}(x,y) - H_{T}(x,y) \right| + \left| \widehat{H}_{n,T}(x',y') - H_{T}(x',y') \right| \right) = o_{P} \left(T^{ \xi - \frac{1}{2}} \right).
\end{align*}
Since $\xi \in(0,1/3)$, the proof is complete.
\end{proof}

\begin{proof}[Proof of Corollary \ref{cor2}:] By an algebraic manipulation, we note that
\begin{align*}
\left|c_{n,T}(u, v) - c(u,v) \right| &\leq \frac{1}{Th^{2}} \int_{0}^{T} \left|K \left( \frac{ \widehat{F}_{n,T}( \hat{V}_{t}^{X})-u}{h}, \frac{ \widehat{G}_{n,T}( \hat{V}_{t}^{Y})-v}{h} \right) - K \left( \frac{F_{T}({V}_{t}^{X})-u}{h}, \frac{G_{T}({V}_{t}^{Y})-v}{h} \right) \right| \mathrm{d}t \\
&+\frac{1}{Th^{2}} \int_{0}^{T} \left|K \left( \frac{F_{T}({V}_{t}^{X})-u}{h}, \frac{G_{T}({V}_{t}^{Y})-v}{h} \right) - K \left( \frac{F({V}_{t}^{X})-u}{h}, \frac{G({V}_{t}^{Y})-v}{h} \right) \right| \mathrm{d}t \\
&+ \left| \frac{1}{Th^{2}} \int_{0}^{T}K \left( \frac{F({V}_{t}^{X})-u}{h}, \frac{G({V}_{t}^{Y})-v}{h} \right) \mathrm{d}t - c(u,v) \right|.
\end{align*}
First, we get by \eqref{equation:OrderOfFV} that
\begin{equation} \label{parta}
\begin{aligned}
&\sup_{t \in [0, \infty)} \left|K \left( \frac{ \widehat{F}_{n,T}( \hat{V}_{t}^{X})-u}{h}, \frac{ \widehat{G}_{n,T}( \hat{V}_{t}^{Y})-v}{h} \right) - K \left( \frac{F_{T}({V}_{t}^{X})-u}{h}, \frac{G_{T}({V}_{t}^{Y})-v}{h} \right) \right| \\
&\leq \frac{C}{h} \left( \sup_{t \in [0, \infty)} \left| \widehat{F}_{n,T}( \hat{V}_{t}^{X}) - F_{T}({V}_{t}^{X}) \right| + \sup_{t \in [0, \infty)} \left| \widehat{G}_{n,T}( \hat{V}_{t}^{Y}) - G_{T}({V}_{t}^{Y}) \right| \right) \\
&= O_{P} \left( \frac{d_{n}}{h} \right).
\end{aligned}
\end{equation}
Since $\sqrt{T}(F_T(x)-F(x))$ converges to a Brownian bridge by Theorem  \ref{thm2}, we can readily deduce that $\sup_{x \in[0, \infty)} \sqrt{T} \big(F_{T}(x) - F(x) \big) = O_{P}(1)$. This yields
\begin{equation} \label{partb}
\begin{aligned}
&\sup_{t \in [0, \infty)} \left|K \left( \frac{F_{T}({V}_{t}^{X})-u}{h}, \frac{G_{T}({V}_{t}^{Y})-v}{h} \right) - K \left( \frac{F({V}_{t}^{X})-u}{h}, \frac{G({V}_{t}^{Y})-v}{h} \right) \right| \\
&\leq \frac{C}{h} \left( \sup_{x \in [0, \infty)} \left|F_{T}(x) - F(x) \right| + \sup_{y \in [0, \infty)} \left|G_{T}(y) - G(y) \right| \right) \\
&= O_{P} \left( \frac{1}{ \sqrt{T}h} \right).
\end{aligned}
\end{equation}
Moreover, as $\big(F(V_{t}^{X}), G(V_{t}^{Y}) \big) \in [0,1]^{2}$ has a density $c(u,v)$:
\begin{align*}
\mathbb{E} \left[ \frac{1}{h^{2}}K \left( \frac{F({V}_{t}^{X})-u}{h}, \frac{G({V}_{t}^{Y})-v}{h} \right) \right] &= \frac{1}{h^{2}} \int_{0}^{1} \int_{0}^{1} K \left( \frac{w-u}{h}, \frac{z-v}{h} \right)c(w,z) \mathrm{d}w \mathrm{d}z \\
&= \int_{0}^{1} \int_{0}^{1} K \left(p, q \right)c(u+ph, v+qh) \mathrm{d}p \mathrm{d}q \\
&= c(u,v) + O_{P}(h).
\end{align*}
Hence, in the above meaning $\frac{1}{Th^{2}} \int_{0}^{T} K \left( \frac{F({V}_{t}^{X})-u}{h}, \frac{G({V}_{t}^{Y})-v}{h} \right) \mathrm{d}t$ is an asymptotically unbiased estimator of $c(u,v)$. Following this line of reasoning, we also conclude that
\begin{equation*}
\mathbb{E} \left[ \frac{1}{h^{2}}K \left( \frac{F({V}_{t}^{X})-u}{h}, \frac{G({V}_{t}^{Y})-v}{h} \right) \cdot \frac{1}{h^{2}}K \left( \frac{F({V}_{s}^{X})-u}{h}, \frac{G({V}_{s}^{Y})-v}{h} \right) \right] = c_{|t-s|}(u,v,u,v) + O_{P}(h).
\end{equation*}
Here, $c_{t}( \cdot, \cdot, \cdot, \cdot)$ is the joint density function of $(F({V}_{0}^{X}), G({V}_{0}^{Y}), F({V}_{t}^{X}), G({V}_{t}^{X}))$. Therefore, we obtain
\begin{equation} \label{partc}
\begin{aligned}
\mathrm{var} \left[ \frac{1}{Th^{2}} \int_{0}^{T} K \left( \frac{F({V}_{t}^{X})-u}{h}, \frac{G({V}_{t}^{Y})-v}{h} \right) \mathrm{d}t \right] &= \frac{1}{T} \int_{0}^{T} \left(c_{t}(u,v,u,v) - c(u,v)^{2} \right) \mathrm{d}t + O_{P}(h) \\
&= O_{P} \left( \frac{1}{T}+h \right).
\end{aligned}
\end{equation}
The latter is obtained by Assumption \ref{asu5}, which implies a Castellana-Leadbetter-type condition \citep[see, for instance, Equation (2.1) of][]{dehay:05a}. In turn, this establishes that
\begin{equation*}
\int_{0}^{T} \left(c_{t}(u,v,u,v) - c(u,v)^{2} \right) \mathrm{d}t = O_{P}(1).
\end{equation*}
Combining \eqref{parta} -- \eqref{partc}, we get
\begin{eqnarray*}
\left|c_{n,T}(u,v) - c(u,v) \right| = O_{P} \left( \frac{d_{n}}{h^{3}} \vee \frac{1}{ \sqrt{T}h^{3}} \vee \left( \frac{1}{ \sqrt{T}} + \sqrt{h} \right) \right).
\end{eqnarray*}
The proof is complete by the chosen order of $h$ and $T$ (relative to $\Delta_{n}$).
\end{proof}

\begin{proof}[Proof of Corollary \ref{corband}:]
First, Theorem \ref{thm2} shows the the functional convergence
\begin{equation*}
\sqrt{T} \big( \widehat{C}_{n,T}(u,v) - C(u,v) \big) \Rightarrow \mathcal{C} \quad \text{in} \quad \big( \mathbb{D}([0,1]^{2}), \norm{ \cdot }_{ \infty} \big),
\end{equation*}
where $\norm{f}_{ \infty} = \sup_{(u,v) \in [0,1]^{2}} |f(u,v)|$ and $\mathcal{C}$ is the centered Gaussian process defined in \eqref{definition:C} with almost surely continuous sample paths on $[0,1]^2$.

Now, consider the mapping $\Psi \colon \mathbb{D}([0,1]^{2}) \longrightarrow \mathbb{R}$:
\begin{equation*}
\Psi(f) = \sup_{(u,v) \in[0,1]^{2}} |f(u,v)|.
\end{equation*}
For any $f,g \in \mathbb{D}([0,1]^2)$,
\begin{equation*}
\bigl| \Psi(f)- \Psi(g) \bigr| = \Bigl| \sup_{(u,v)}|f(u,v)| - \sup_{(u,v)}|g(u,v)|  \Bigr| \leq \sup_{(u,v)}|f(u,v) - g(u,v)| = \|f-g \|_{ \infty}.
\end{equation*}
Hence, $\Psi$ is Lipschitz continuous in this metrized space.

Since $\sqrt{T}( \widehat{C}_{n,T}(u,v) - C(u,v)) \Rightarrow \mathcal{C}$ in $\big( \mathbb{D}([0,1]^{2}), \norm{ \cdot}_{ \infty} \big)$ and $\Psi$ is continuous on this space, the continuous mapping theorem yields
\begin{equation*}
\Psi \big( \sqrt{T}( \widehat{C}_{n,T}(u,v) - C(u,v)) \big) \Rightarrow \Psi( \mathcal{C}).
\end{equation*}
That is,
\begin{equation*}
\sup_{(u,v) \in[0,1]^{2}} \sqrt{T} \big( \widehat{C}_{n,T}(u,v) - C(u,v) \big) \overset{d}{ \longrightarrow} \sup_{(u,v) \in[0,1]^{2}} \mathcal{C}(u,v),
\end{equation*}
as $\Delta_{n} \rightarrow 0$ and $T \rightarrow \infty$. This establishes \eqref{equation:sup}.
\end{proof}

\begin{proof}[Proof of Proposition \ref{confidenceband}:]

We divide the proof in three parts. \\

(1) \textit{Finite-dimensional Gaussian approximation}: We start by fixing a finite grid $\{(u_{i}, v_{j}) : 1 \leq i,j \leq m \} \subset[0,1]^{2}$ and define
\begin{equation*}
Z_{n,T}(u_{i},v_{j}) = \sqrt{T} \big( \widehat{C}_{n,T}(u_{i},v_{j}) - C(u_{i},v_{j}) \big).
\end{equation*}
together with the corresponding $m^{2}$-dimensional random vector
\begin{equation*}
Z_{n,T}^{(m)} = \big(Z_{n,T}(u_{i},v_{j}) \big)_{1 \leq i,j \leq m}.
\end{equation*}
By Theorem \ref{thm2} and the continuous mapping theorem, as $\Delta_{n} \rightarrow 0$ and $T \rightarrow \infty$,
\begin{equation*}
Z_{n,T}^{(m)} \Rightarrow Z^{(m)},
\end{equation*}
where $Z^{(m)} = \big( \mathcal{C}(u_{i},v_{j}) \big)_{1 \leq i,j \leq m}$ is a centered Gaussian vector with covariance matrix
\begin{equation*}
\Sigma_{m} = \big\{ \mathrm{avar}_{ \mathcal{C}}(u_{i}, v_{j}, u_{k}, v_{l}) \big\}_{1 \leq i,j,k,l \leq m}.
\end{equation*}
By Corollary \ref{cor1},
\begin{equation*}
\widehat{ \Sigma}_{m} = \Big\{ \widehat{ \mathrm{avar}}_{ \mathcal{C}}(u_{i},v_{j},u_{k},v_{l}) \Big\}_{1 \leq i,j,k,l \leq m} \overset{ \mathbb{P}}{ \longrightarrow} \Sigma_{m}.
\end{equation*}
Conditional on $\widehat{C}_{ \mathcal{C}}$, we generate
\begin{equation*}
\widehat{Z}_{n,T}^{(m)} = \big( \widehat{ \mathcal{C}}_{n,T}(u_{i},v_{j}) \big)_{1 \leq i,j \leq m} \sim N \big(0, \widehat{ \Sigma}_{m} \big).
\end{equation*}
Since $\widehat{ \Sigma}_{m} \overset{ \mathbb{P}}{ \longrightarrow} \Sigma_{m}$, the characteristic functions of $N(0, \widehat{ \Sigma}_{m})$ converge pointwise to those of
$N(0, \Sigma_{m})$. Hence, by L\'{e}vy's continuity theorem,
\begin{equation*}
\mathcal{L} \big( \widehat{Z}_{n,T}^{(m)} \mid \widehat{C}_{ \mathcal{C}} \big) = \mathcal{L} \big(N(0, \widehat{ \Sigma}_{m}) \big) \overset{ \mathbb{P}}{ \longrightarrow} \mathcal{L} \big(N(0, \Sigma_{m}) \big) = \mathcal{L} \big(Z^{(m)} \big),
\end{equation*}
In words, the conditional law of $\widehat{Z}_{n,T}^{(m)}$ converges in
probability to that of $Z^{(m)}$.

Let $\phi_{m} : \mathbb{R}^{m^{2}} \rightarrow \mathbb{R}$ be the functional
$\phi_{m}(z) = \max_{1 \leq i,j \leq m}|z_{ij}|$. By the continuous mapping theorem, for any continuity point $c>0$ of the distribution function of $\phi_{m}(Z^{(m)})$,
\begin{equation*}
\mathbb{P} \Big( \phi_{m}( \widehat{Z}_{n,T}^{(m)}) \leq c \mid \widehat{C}_{ \mathcal{C}} \Big) \overset{ \mathbb{P}}{ \longrightarrow} \mathbb{P} \big( \phi_{m}(Z^{(m)}) \leq c \big).
\end{equation*}
Or, equivalently,
\begin{equation*}
F_{n,T}^{(m)}(c \mid \widehat{C}_{ \mathcal{C}}) \overset{ \mathbb{P}}{ \longrightarrow} F_{m}(c),
\end{equation*}
where
\begin{equation*}
F_{n,T}^{(m)}(c \mid \widehat{C}_{ \mathcal{C}}) = \mathbb{P} \left( \sup_{1 \leq i,j \leq m} \big| \widehat{ \mathcal{C}}_{n,T}(u_{i},v_{j}) \big| \leq c \mid \widehat{C}_{ \mathcal{C}} \right) \quad \mathrm{and} \quad F_{m}(c) = \mathbb{P} \left( \sup_{1 \leq i,j \leq m}| \mathcal{C}(u_{i},v_{j})| \leq c \right).
\end{equation*}

(2) \textit{Quantile consistency}:
We now let $m \rightarrow \infty$ such that the grid becomes dense. Then, since $\mathcal{C}$ has continuous sample paths on $[0,1]^{2}$, almost surely, it follows that
\begin{equation*}
\sup_{1 \leq i,j \leq m}| \mathcal{C}(u_{i},v_{j})| \rightarrow
\sup_{(u,v) \in[0,1]^{2}}| \mathcal{C}(u,v)|,
\end{equation*}
almost surely, so the sequence of distribution functions $(F_{m})_{m=1}^{ \infty}$ converge pointwise to
\begin{equation*}
\lim_{m \rightarrow \infty} F_{m}(c) = F(c) = \mathbb{P} \left( \sup_{(u,v) \in[0,1]^{2}}| \mathcal{C}(u,v)| \leq c \right),
\end{equation*}
at every continuity point $c$ of $F$. Hence, for any fixed compact interval $[0,M]$ and $\varepsilon > 0$, we can first choose $m$ large enough so that
\begin{equation*}
\sup_{c \in[0,M]}|F_{m}(c) - F(c)| < \varepsilon,
\end{equation*}
and then choose $(n,T)$ sufficiently large so that for every $\varepsilon, \eta > 0$,
\begin{equation*}
\mathbb{P} \left( \sup_{c \in [0,M]} \Big|F_{n,T}^{(m)}(c \mid \widehat{C}_{ \mathcal{C}}) - F_{m}(c) \Big| > \varepsilon \right) < \eta.
\end{equation*}
Combining the two inequalities yields
\begin{equation*}
\sup_{c \in[0,M]} \Big|F_{n,T}^{(m)}(c \mid \widehat{C}_{ \mathcal{C}}) - F(c) \Big| \overset{ \mathbb{P}}{ \longrightarrow} 0.
\end{equation*}
Hence, $F_{n,T}^{(m)}( \cdot \mid \widehat{C}_{ \mathcal{C}}) \overset{ \mathbb{P}}{ \longrightarrow} F(\cdot)$ uniformly on compact subsets of $(0, \infty)$.

To complete part (b), let
\begin{equation*}
q_{1- \alpha}^{CB} = \inf \{c>0:F(c) \geq 1- \alpha \} \quad \mathrm{and} \quad
\widehat{q}_{1- \alpha}^{CB} = \inf \{c>0:F_{n,T}^{(m)}(c \mid \widehat{C}_{ \mathcal{C}}) \geq 1- \alpha \}
\end{equation*}
denote the unique $(1- \alpha)$-quantile of $F( \cdot)$ and $F_{n,T}^{(m)}(\cdot \mid \widehat{C}_{ \mathcal{C}})$, respectively. The uniform convergence of $F_{n,T}$ to $F$ and continuity of $F$ at $q_{1- \alpha}^{CB}$ imply that
\begin{equation*}
\widehat{q}_{1- \alpha}^{CB} \overset{ \mathbb{P}}{ \longrightarrow} q_{1- \alpha}^{CB},
\end{equation*}
see \citet[][Lemma 21.2]{vaart:98a}.

(3) \textit{Asymptotic coverage of the confidence band}: In the final part, we set
\begin{equation*}
\mathcal{X}_{n,T} = \sup_{(u,v) \in [0,1]^{2}} \sqrt{T} \big( \widehat{C}_{n,T}(u,v) -C(u,v) \big) \quad \mathrm{and} \quad \mathcal{X} = \sup_{(u,v) \in [0,1]^{2}} \mathcal{C}(u,v).
\end{equation*}
By Corollary \ref{corband}, $\mathcal{X}_{n,T} \Rightarrow \mathcal{X}$ as $\Delta_{n} \rightarrow 0$ and $T \rightarrow \infty$, and furthermore $\widehat{q}_{1- \alpha}^{CB} \overset{ \mathbb{P}}{ \longrightarrow} q_{1- \alpha}^{CB}$ with
$\mathbb{P}( \mathcal{X}X \leq q_{1- \alpha}^{CB}) = 1- \alpha$. An application of Slutsky's lemma then delivers
\begin{equation*}
\mathbb{P} \left( \mathcal{X}_{n,T} \leq \widehat{q}_{1- \alpha}^{CB} \right) \rightarrow \mathbb{P} \left( \mathcal{X} \leq q_{1- \alpha}^{CB} \right) = 1- \alpha.
\end{equation*}
Equivalently,
\begin{equation*}
\mathbb{P} \left( \sup_{(u,v) \in[0,1]^{2}} \bigl| \widehat{C}_{n,T}(u,v) - C(u,v) \bigr| \leq \widehat{q}_{1- \alpha}^{CB} / \sqrt{T} \right) \rightarrow 1 -\alpha,
\end{equation*}
which shows that the uniform confidence band
\begin{equation*}
\left\{C: \sup_{(u,v) \in[0,1]^{2}} \bigl| \widehat{C}_{n,T}(u,v) - C(u,v) \bigr|
\leq \widehat{q}_{1- \alpha}^{CB} / \sqrt{T} \right\}
\end{equation*}
has asymptotic coverage $1- \alpha$.
\end{proof}

\begin{proof}[Proof of Corollary \ref{cortest}:]

Let
\begin{equation*}
G_{n,T}(u,v) = \sqrt{T} \big( \widehat{C}_{n,T}(u,v) - C_{0}(u,v) \big)
\qquad (u,v) \in [0,1]^{2},
\end{equation*}
be the empirical process associated with $\mathcal{H}_{0}: C = C_{0}$. By Theorem \ref{thm2},
\begin{equation*}
G_{n,T} \Rightarrow \mathcal{C} \quad \text{in} \quad \big( \mathbb{D}([0,1]^{2}), \norm{ \cdot}_{ \infty} \big),
\end{equation*}
under $\mathcal{H}_{0}$, where $\mathcal{C}$ is the centered Gaussian process from Corollary \ref{corband}, which has almost surely continuous, and bounded, sample paths on $[0,1]^{2}$.

Now, consider the functional $\Phi: \mathbb{D}([0,1]^{2}) \rightarrow \mathbb{R}_{+}$ defined as
\begin{equation*}
\Phi(f) = \int_{0}^{1} \int_{0}^{1} f^{2}(u,v) \mathrm{d}u \mathrm{d}v.
\end{equation*}
For any sequence $(f_{n})_{n=1}^{ \infty}$ such that $f_{n} \rightarrow  f$ in $\big( \mathbb{D}([0,1]^2), \norm{ \cdot}_{ \infty} \big)$:
\begin{equation*}
\sup_{(u,v) \in [0,1]^{2}} \big|f_{n}(u,v)^{2} - f(u,v)^{2} \big| \leq 2 \bigg( \sup_{(u,v)}|f_{n}(u,v)| + \sup_{(u,v)}|f(u,v)| \bigg) \, \sup_{(u,v)}|f_{n}(u,v) - f(u,v)|.
\end{equation*}
Since $\norm{f_{n}}_{ \infty} \rightarrow \norm{f}_{ \infty} < \infty$, the right-hand side converges to zero, and so $f_{n}^{2} \rightarrow f^{2}$ uniformly on $[0,1]^2$. It thus follows that
\begin{equation*}
\lim_{n \rightarrow \infty} \Phi(f_{n}) = \lim_{n \rightarrow \infty} \int_{0}^{1} \int_{0}^{1} f_{n}^{2}(u,v) \mathrm{d}u \mathrm{d}v = \int_{0}^{1} \int_{0}^{1} f^{2}(u,v) \mathrm{d}u \mathrm{d}v =\Phi(f),
\end{equation*}
showing that $\Phi$ is continuous on $\big( \mathbb{D}([0,1]^2), \norm{ \cdot}_{ \infty} \big)$.

So, by the continuous mapping theorem
\begin{equation*}
\Phi \big(G_{n,T} \big) = \int_{0}^{1} \int_{0}^{1}G_{n,T}^{2}(u,v) \mathrm{d}u \mathrm{d}v \overset{d}{ \longrightarrow} \int_{0}^{1} \int_{0}^{1} \mathcal{C}^{2}(u,v) \mathrm{d}u \mathrm{d}v,
\end{equation*}
as $\Delta_{n} \rightarrow 0$ and $T \rightarrow \infty$. This establishes \eqref{eq:limitsup}.
\end{proof}
\pagebreak

% LIST OF REFERENCES

\renewcommand{\baselinestretch}{1.0}
\normalsize
\bibliographystyle{rfs}
\bibliography{userref}

\end{document}